\documentclass[12pt]{article}

\usepackage{graphicx}
\usepackage{amssymb}
\usepackage{color}

\usepackage{mathrsfs}
\usepackage{amsmath}
\usepackage{amsfonts}
\usepackage[colorlinks,linkcolor=black,anchorcolor=blue,citecolor=blue]{hyperref}
\usepackage{amssymb, amsmath, cite}

\newtheorem{theorem}{Theorem}[section]
\newtheorem{lemma}{Lemma}[section]

\newenvironment{proof}{\smallskip\noindent{\bf Proof.}\hskip \labelsep}
\newcommand\be{\begin{equation}}
\newcommand\ee{\end{equation}}
\newcommand\ber{\begin{eqnarray}}
\newcommand\eer{\end{eqnarray}}
\newcommand\bea{\begin{eqnarray}}
\newcommand\eea{\end{eqnarray}}
\newcommand\berr{\begin{eqnarray*}}
\newcommand\eerr{\end{eqnarray*}}
{\setlength\arraycolsep{2pt}

\newcommand\e{\mathrm{e}}\newcommand\pa{\partial}
\newcommand\ii{\mathrm{i}}
\newcommand{\dd}{\mathrm{d}}\newcommand{\vp}{\varphi}\newcommand{\D}{{\cal D}}
\newcommand{\nn}{\nonumber}\newcommand\ep{\epsilon}
\newcommand{\vep}{\varepsilon}\newcommand\bfR{\mathbb{R}}
\newcommand\lm{\lambda}\newcommand{\HH}{\mathscr{H}}
\setlength{\baselineskip}{1.3\baselineskip}

\begin{document}

\title{Domain Wall Solitons Arising in \\ Classical Gauge Field Theories}
\author{Lei Cao\\School of Mathematics and Statistics\\ Henan University\\
 Kaifeng, Henan 475004, PR China\\\\ Shouxin Chen\\Institute of Contemporary Mathematics, and\\School of Mathematics and Statistics\\Henan University\\Kaifeng, Henan 475004, PR China\\\\Yisong Yang\\Courant Institute of Mathematical Sciences\\New York University\\New York, NY 10012, USA}
\date{}
\maketitle

\begin{abstract} Domain wall solitons are basic constructs realizing phase transitions in various field-theoretical models and are solutions to 
some nonlinear ordinary differential equations descending from the corresponding full sets of governing equations in higher dimensions. In this paper,
we present a series of domain wall solitons arising in several classical gauge field theory models. In the context of the Abelian gauge field theory, we 
unveil the surprising result that the solutions may explicitly be constructed, which enriches our knowledge on integrability of the planar Liouville type equations
in their one-dimensional limits. In the context of the non-Abelian gauge field theory, we obtain some existence theorems for domain wall solutions arising in the 
 electroweak type theories by developing some methods of calculus of variations formulated as direct and constrained minimization problems over a weighted Sobolev space.
%\medskip




\end{abstract}

\section{Introduction}\label{s1}

A domain-wall soliton describes a phase transition between two phases often referred to as domains. In the study of magnetism,
a magnetic domain is a region within a material which has uniform magnetization so that the individual moments of the atoms are aligned with one another.
A domain wall is an interface separating two distinct magnetic domains, realizing a transition between different magnetic moments usually described by an angular displacement variable.
One of the earliest domain wall models is that given by the classical one spatially dimensional
sine--Gordon equation which governs the simplest domain walls occurring in the Landau--Lifshitz magnetism theory. 
Being integrable, the sine--Gordon model has served
as an illustrative mathematical laboratory and inspired the development of many profound concepts in classical and quantum field theories \cite{DE,FK,R}. Unfortunately, due to the
complicated structures of domain wall models,
most of the underlying governing equations are not integrable, and thus, one needs to resort to nonlinear functional analysis in order to achieve some understanding
of the problem of interest. For example, in the Ginzburg--Landau theory of superconductivity \cite{A,GL,S,T},
the energy of a domain wall soliton connecting the normal and superconducting phases, called the surface energy, is to be evaluated, in order to classify the types of superconductors, so that a negative surface energy results in Type I superconductors, and a positive one, Type II. However, a mathematically rigorous classification
of superconductivity along the line of determining the sign of surface energy has not been established due to the difficulties involved in an existence theory for the domain wall solutions. As another example, we mention a sweeping mechanism proposed in \cite{DLV} to resolve the celebrated monopole problem \cite{P,Pre1} via domain walls
which may be created along with monopoles but the latter can be annihilated by the former through unwinding and dissipation. These and many other problems
suggest that various domain wall equations arising in field theories are of interest and importance to be studied systematically. 

It is well known that the equations of motion of gauge field theories in their general setting are difficult to approach except in a critical situation when there is
a so-called BPS structure after the work of Bogomol'nyi \cite{B} and Prasad--Sommerfield \cite{PS}.
The exploration of such a BPS structure has enabled a harvest of knowledge on topological solitons, including vortices, monopoles, and instantons in the past, and
more recently, domain walls as well, by virtue of superpotentials (cf. \cite{CY} and references therein)
and nonlinear functional analysis (cf. \cite{Sakai} as an example). In the present work, we construct BPS domain wall solitons
in some simplest but classical and fundamental gauge field theories.

First, we consider the Abelian Higgs and Abelian Chern--Simons--Higgs theories for which the BPS vortex equations have been well studied. The common feature of these
equations is that they may all be reduced into a class of nonlinear elliptic equations of the Liouville type but,
unlike the Liouville equation, the broken vacuum symmetry renders these equations
non-integrable \cite{Schiff}. It is interesting that these equations contain a domain wall substructure. We show that the underlying domain wall equations,
which are also of the Liouville type, allow a certain level of integration so that their exact solutions may well be described. More interestingly, in the Abelian Chern--Simons--Higgs context,
the equation can be completely integrated to give us its full family of solutions explicitly. As a special by-product, we observe that the Abelian Chern--Simons--Higgs equation
is ``more integrable" than the simpler Abelian Higgs equation, which seems surprising and rather unexpected. This part of investigation consists of Sections 2 and 3.

Next, we consider the domain wall equations arising in the electroweak theory which may be regarded as a simplest non-Abelian Yang--Mills--Higgs theory.
As before, these equations are contained in the BPS vortex equations as dimensionally reduced substructures and of the Liouville type. The added subtlety of
the problem is that the governing equations are coupled systems of nonlinear equations which are often non-integrable except in isolated special situations.
Here, we shall resort to means of global functional analysis to construct domain wall solutions. Specifically, we will formulate some variational methods to solve
the electroweak domain wall equations derived in \cite{Bol} and in \cite{AO2,AO3,AO4}  subject to their corresponding boundary conditions. In order to do so, we need to introduce a weighted Sobolev space and develop
appropriate space-embedding tools characterized by a Trudinger--Moser type inequality \cite{Aubin,M,Trudinger}, which enables a direct minimization method
and a constrained minimization method to be effectively carried out. This part of the work consists of Sections 4--7 and will be useful for solving other more
complicated domain wall equations

In the last section, we conclude the paper.

\section{Domain wall equations in Abelian Higgs theory}\label{s2}
\setcounter{equation}{0}\setcounter{remark}{0}

Following Bolognesi {\em et al} \cite{Bol}, we first consider the classical Abelian Higgs theory \cite{JT,NO} defined by the action density
\be\label{x2.1}
{\cal L}=-\frac14 F_{\mu\nu}F^{\mu\nu}+D_\mu q \overline{D^\mu q} -\frac{\kappa^2e^2}2(|q|^2-\xi)^2,
\ee
where $F_{\mu\nu}=\pa_\mu A_\nu -\pa_\nu A_\mu$, $D_\mu q=\pa_\mu q-\ii eA_\mu q$, $\mu,\nu=0,1,2,3$, and $e,\kappa,\xi$ are positive coupling parameters,
so that $\kappa=1$ spells out the BPS limit \cite{B,PS}. With $q=\e^{\ii E t}\phi(x^1,x^2)$ $A_0=\frac{E}e$, $A_3=0$, $A_1,A_2$ being functions depending
on $x^1,x^2$ only, and $\kappa=1$, the Hamiltonian of (\ref{x2.1}) takes the form
\bea
{\cal H}&=&\frac12 F_{12}^2+|D_1\phi|^2+|D_2\phi|^2+\frac{e^2}2(|\phi|^2-\xi)^2\nn\\
&=&\frac12\left(F_{12}+e(|\phi|^2-\xi)\right)^2+|D_1\phi+\ii D_2\phi|^2+e\xi   F_{12}-\ii \epsilon^{ij}\pa_i\left(\overline{\phi}D_j\phi\right),
\eea
where $i,j=1,2$, which leads to the Euler--Lagrange equations
\be\label{x2.3a}
D_i^2\phi=e^2(|\phi|^2-\xi)\phi,\quad \pa_j F_{ij}=\ii e(\phi \overline{D_i\phi}-\overline{\phi}D_i\phi),
\ee
 and the associated BPS vortex equations
\be\label{x2.3}
D_1\phi+\ii D_2\phi =0,\quad
F_{12}=e(\xi-|\phi|^2),\quad x\in\bfR^2,
\ee
respectively, so that (\ref{x2.3}) implies (\ref{x2.3a}). In \cite{JT,T2}, it is shown that (\ref{x2.3a}) and (\ref{x2.3}) are actually equivalent over solutions of
finite energies.
Setting $u=\ln|\phi|^2$, the system (\ref{x2.3}) is reduced \cite{JT}  into the following Liouville type equation \cite{L}
\be\label{x2.4}
\Delta u=e^2(\e^u-\xi)+4\pi\sum_{s=1}^N\delta_{p_s}(x),
\ee
where $\delta_p$ is the Dirac distribution concentrated at $p\in\bfR^2$. An existence and uniqueness theorem for a solution 
of (\ref{x2.4}) satisfying $u=\ln \xi$ at infinity has been established by Taubes \cite{JT,T1}. An open problem in \cite{JT} asks whether one can obtain the solution of (\ref{x2.4}) explicitly. In \cite{Schiff}, it is shown
that (\ref{x2.4})  is not integrable whenever $\xi\neq0$. In fact, the Liouville equation \cite{L} corresponds to the situation where $\xi=0$ in (\ref{x2.4}) which is
integrable by all known methods of integration, including Liouville's method \cite{L}, the B\"{a}cklund transformation \cite{McC}, the method of separation of variables \cite{Leznov},
and inverse scattering \cite{Andreev}.

As in \cite{Bol}, note that it is consistent to impose the domain-wall ansatz in (\ref{x2.3a}) and (\ref{x2.3}) with
\be\label{x2.5}
\phi=\phi(x)=\mbox{real}, \quad A_1=0,\quad A_2=A(x),\quad x=x^1,
\ee
which reduces (\ref{x2.3a})  into
\be\label{x2.6a}
\phi''-e^2 A^2\phi=e^2(\phi^2-\xi)\phi,\quad A''=2e^2\phi^2 A,
\ee
which are the celebrated one-dimensional Ginzburg--Landau equations \cite{S,T}, and (\ref{x2.3}) into
\be\label{x2.6}
\phi'+eA\phi=0,\quad A'=e(\xi-\phi^2),\quad -\infty<x<\infty,
\ee
respectively. It may be examined that (\ref{x2.6}) implies (\ref{x2.6a}). That is, as in their two-dimensional settings, (\ref{x2.3a}) and (\ref{x2.3}), (\ref{x2.6})
is a reduction of (\ref{x2.6a}), which will now be considered.

To proceed, we see from the first equation in (\ref{x2.6}) that $\phi$ cannot change sign so that we may assume $\phi>0$ without loss of generality for a nontrivial
solution. Hence the first equation in (\ref{x2.6}) gives us $(\ln \phi)'+eA=0$. Inserting this into the second equation in (\ref{x2.6}) and setting $u=\ln\phi^2-\ln\xi$, we arrive at the following normalized equation
\be\label{x2.7}
u''=\lm(\e^u-1),\quad-\infty<x<\infty,
\ee
of a Liouville type, where $\lm=2e^2\xi$. By the maximum principle, it is easily seen that the only solution of (\ref{x2.7}) satisfying the boundary condition
$u(\pm\infty)=0$ is the trivial one, $u=0$. In \cite{Bol}, it is shown with numerics that (\ref{x2.7}) has nontrivial solutions satisfying the boundary conditions
\be\label{x2.8}
u(-\infty)=0,\quad  u(\infty)=-\infty, 
\ee
and 
\be\label{x2.9}
u(-\infty)=-\infty,\quad u(\infty)=-\infty,
\ee
respectively. The phase $u=0$ corresponds to $\phi^2=\xi$ which would give rise to a vanishing magnetic field, $A'=0$, as a consequence of the Meissner effect
\cite{S,T}. This phase may be referred to as the Higgs phase since it pertains to a spontaneous broken vacuum symmetry. The phase $u=-\infty$, on the other hand,
corresponds to $\phi=0$, which would give rise to a full presence of the magnetic field, $A'=e\xi$. This phase may be referred to as the magnetic phase. See \cite{Bol}. Such a correspondence may also be seen clearly via the BPS equations (\ref{x2.6}). Thus, the solutions of (\ref{x2.7}) subject to the boundary conditions
(\ref{x2.8}) and (\ref{x2.9}) are domain-wall solitons representing transitions between these phases.
Below we construct such solutions by integrating (\ref{x2.7}).

To begin with, note that the boundary condition $u(-\infty)=0$ implies $u'(-\infty)=0$. Thus, multiplying (\ref{x2.7}) by $u'$ and integrating over $(-\infty,x)$, we arrive at
the Friedmann type equation
\be\label{x2.10}
(u')^2=2\lm (\e^u-u-1).
\ee
The right-hand side of (\ref{x2.10}) is  positive for any $u\neq0$. So we may rewrite (\ref{x2.10}) as
\be\label{x2.11}
u'=\pm \sqrt{2\lm}\sqrt{\e^u-u-1}.
\ee
Since we are interested in a solution satisfying $u(\infty)=-\infty$, we need to choose the lower sign in (\ref{x2.11}). Besides, let $x_0$ be such that $u_0=u(x_0)<-1$.
Then we have in view of (\ref{x2.11}) the integral
\be\label{x2.12}
\int_{u_0}^{u(x)}\frac{\dd u}{\sqrt{\e^u-u-1}}=-\sqrt{2\lm}(x-x_0).
\ee
Using $-u-1<\e^u-u-1<-u$ in (\ref{x2.12}), we get
\be\label{x2.13}
2\left(\sqrt{-u(x)-1}-\sqrt{-u_0-1}\right)>\sqrt{2\lm}(x-x_0)>2\left(\sqrt{-u(x)}-\sqrt{-u_0}\right),\quad x>x_0,
\ee
since $u(x)<u_0$, which leads to the following sharp asymptotic estimate
\be\label{x2.14}
u(x)= -\frac{\lm x^2}2+\mbox{O}(1),\quad x\to\infty.
\ee

In a similar token, we can estimate the behavior of $u(x)$ as $x\to-\infty$. In fact, since $u(-\infty)=0$ and $u'<0$,
for any $\vep\in(0,1)$ we may choose $x_0$ sufficiently negative
such that $\e^{u(x)}>1-\vep$ for $x<x_0$. Hence
\be
(1-\vep) u(x)^2 <2\left(\e^{u(x)}-u(x)-1\right)<u^2(x),\quad x<x_0.
\ee
Inserting this into (\ref{x2.12}), we obtain
\be
\frac1{\sqrt{1-\vep}}\ln\left|\frac{u(x)}{u_0}\right|<\sqrt{\lm}(x-x_0)<\ln\left|\frac{u(x)}{u_0}\right|,\quad x<x_0.
\ee
Since $\vep>0$ may be arbitrarily small, we are led to the following sharp asymptotic estimate
\be\label{xx2.19}
u(x)=\mbox{O}\left(\e^{-\sqrt{\lm}|x|}\right),\quad x\to-\infty.
\ee

It may be shown that, up to a translation, the above constructed solution is unique. In fact, by translation invariance, we may assume $u(0)=-1$ (say). Then we
may use a standard continuity argument to show that $u'(0)$ is uniquely determined \cite{Ylee,Ybook}.

Next, we construct solutions of (\ref{x2.7}) satisfying the boundary condition (\ref{x2.9}).
Such a solution will have a global maximum $u_0$ (say). By translation invariance of (\ref{x2.7}), we may assume $u(0)=u_0$. Since $u''(0)\leq0$, we have $u_0\leq0$ in view of (\ref{x2.7}). In particular, $u(x)\leq0$ for all $x$. Hence, multiplying (\ref{x2.7}) by
$u'$, integrating around $x=0$, and using $u'(0)=0$, we obtain
\be\label{x2.18}
(u')^2 =2\lm(\e^u-u -\e^{u_0}+u_0),
\ee
which is slightly different from (\ref{x2.10}).
Let $f(u)=\e^u-u -\e^{u_0}+u_0$. Then $f(u_0)=0$.  Since $f'(u)<0$ for $u<0$, so $f(u)>0$ for all $u<u_0$. In other words, the right-hand side of (\ref{x2.18})
remains positive for $x\neq0$. Thus we obtain
\be\label{x2.19}
u'=\pm\sqrt{2\lm}\sqrt{\e^u-u-\e^{u_0}+u_0}.
\ee
If the solution goes to $-\infty$ as $x\to \infty$, we need to choose the lower sign in (\ref{x2.19}). Thus we get
\be\label{x2.20}
\int_{u_0}^{u(x)}\frac{\dd u}{\sqrt{\e^u-u-\e^{u_0}+u_0}}=-\sqrt{2\lm}x,\quad x>0.
\ee
For the part in $x<0$, since we are to get $u(-\infty)=-\infty$, we may choose the upper sign in (\ref{x2.19}) or flip the solution in $x>0$ by setting $x\mapsto -x$
to get the solution with $x<0$. In this way we obtain a solution of (\ref{x2.7}) which satisfies the boundary condition (\ref{x2.9}).

Summarizing, we may state the following theorem regarding (\ref{x2.7}).

\begin{theorem}
The one-dimensional Liouville type equation (\ref{x2.7}) governing the domain-wall solitons in the Abelian Higgs theory is partially integrable in the sense that its
solutions may all be constructed via a further reduced first-order equation, (\ref{x2.10}) or (\ref{x2.18}), which gives us the following conclusions.
\begin{enumerate}
\item[(i)] Under the boundary condition (\ref{x2.8}), the equation (\ref{x2.7}) becomes (\ref{x2.10}) whose solution is unique up to a translation and enjoys
the behavior (\ref{x2.14}) and (\ref{xx2.19}) asymptotically, as $x\to\pm\infty$, respectively. In other words, in this situation, the full set of solutions is a 
one-parameter family of functions related by translations.

\item[(ii)] Under the boundary condition (\ref{x2.9}),
for any given $x_0$ and $ u_0\leq0$, the equation (\ref{x2.7})  has a unique solution $u$ which attains its global maximum $u_0$ at $x=x_0$ which is
given by the first-order equation (\ref{x2.18}), or more precisely,
(\ref{x2.19}) with the choices of the upper and lower signs for $x<x_0$ and $x>x_0$, respectively, and enjoys the sharp asymptotic behavior
\be
u(x)=-\frac{\lm x^2}2+\mbox{O}(1),\quad x\to\pm\infty.
\ee
In other words,  a full two-parameter family description is obtained for the solutions of (\ref{x2.7}) subject to the boundary condition (\ref{x2.9}).
\end{enumerate}
\end{theorem}

We may extend our study to the situation of a massive $SU(2)$ gauge field theory known as the Georgi--Glashow model. For this purpose, let $\sigma_a$ ($a=1,2,3$)
denote the Pauli spin matrices.
Then $t_a=\frac{\sigma_a}2, \,a=1,2,3$ is a set of generators of $SU(2)$
satisfying the commutation relation
$
[t_a,t_b]=\ii\ep_{abc}t_c
$
and  Tr$(t_a t_b)=\frac{\delta_{ab}}2$. Any $su(2)$-valued gauge potential $A_\mu$ may be represented
 in the 
matrix form   
$
A_\mu=A^a_\mu t_a.
$
As a result,
the massive, simplest non-Abelian, gauge theory model under consideration governing the dynamics of a
special particle mediating electroweak interactions, called the $W$-particle,  which is represented by the complex field
$
W_\mu=\frac1{\sqrt{2}}(A^1_\mu+\ii A^2_\mu),
$
is defined by the Lagrangian density
\be\label{f1}
{\cal L}=-\frac12 \mbox{ Tr }( F_{\mu\nu} F^{\mu\nu})+m_W^2 \overline{W}_\mu W^\mu,
\ee
over the $(3+1)$-dimensional Minkowski spacetime, where the field strength tensor $F_{\mu\nu}$, is given by
\be\label{f2}
F_{\mu\nu}=F_{\mu\nu}^a t_a=\pa_\mu A_\nu-\pa_\nu A_\mu+\ii e[A_\mu,A_\nu],
\ee
with $m_W>0$ the $W$-particle mass and $e>0$ the electric charge parameter.
It is customary to regard  $A^3_\mu$  as an
electromagnetic gauge potential, $A^3_\mu=P_\mu$,
with the
associated field strength tensor
$
P_{\mu\nu}=\partial_\mu P_\nu-\partial_\nu P_\mu.
$
With the notation $D_\mu=\partial_\mu-\ii e P_\mu$, we have
\bea\label{f3}
{\cal L}
&=&-\frac14 P_{\mu\nu}P^{\mu\nu}-\frac12(D_\mu W_\nu-D_\nu W_\mu)
\overline{(D^\mu W^\nu-D^\nu W^\mu)}
  +m^2_W \overline{W}_\mu W^\mu\nn\\
&&+\ii e P_{\mu\nu}W^{\mu}\overline{W}^\nu+\frac{e^2}2\,([\overline{W}_\mu
\overline{W}^\mu][W_\nu W^\nu]-[\overline{W}_\mu W^\mu]^2),
\eea
which describes the interaction between the weak force bosons, or the $W$-particles,
 and the electromagnetic
photons.
Vortex-like solutions are characterized by the further ansatz
$
W_0=W_3=0,  P_0=P_3=0, 
W_j, P_j$ depending only 
on $x^1,x^2$, $ j=1,2$, so that the equations of motion are greatly reduced. Furthermore, it is consistent to assume that $W_1, W_2$  be represented by a 
single complex scalar field $W$ through
$
W_1=W, W_2=\ii W.
$
As a consequence, the equations of motions may be derived as the Euler--Lagrange equations of the Hamiltonian
\begin{equation}\label{f8}
{\cal H}=\frac12P^2_{12}+|D_1W+\ii D_2W|^2+2m^2_W|W|^2-2eP_{12}|W|^2+2e^2|W|^4,
\end{equation}
which are
\bea\label{f9}
D_i^2W&=&2m^2_WW-3eP_{12}W+4e^2|W|^2W,\\
\partial_i P_{ij}&=&\ii e(\overline{W}[D_jW]-W\overline{[D_jW]})
                 +3e\ep_{ij}(\overline{W}[
D_iW]+W\overline{[D_iW]}).\label{f10}
\eea
These equations are still rather complicated.  Fortunately, it is derived by Ambj{\o}rn and Olesen in \cite{AO1}  the BPS vortex equations
\be\label{x2.22}
D_1 W+\ii D_2 W=0,\quad P_{12}=2e|W|^2+\frac{m^2_W}e,
\ee
so that any solution of (\ref{x2.22}) solves (\ref{f9}) and (\ref{f10}) as well. Here we note further that these equations also allow
a domain wall substructure. In  fact, as in (\ref{x2.5}), if we impose the domain-wall  ansatz
\be
W=W(x)=\mbox{real}, \quad P_1=0,\quad P_2=P(x),
\ee
then (\ref{f9})--(\ref{f10}) and (\ref{x2.22}) become
\be\label{x232}
W''-e^2 P^2 W=2m_W^2 W-3e P' W+4e^2 W^3,\quad P''=2e^2 W^2 P+6e W W',
\ee
and
\be\label{x2.23}
W'+ePW=0,\quad P'=2e W^2+\frac{m_W^2}e,
\ee
respectively. It can be checked that any solution of (\ref{x2.23}) satisfies (\ref{x232}). That is, as in the Abelian Higgs situation, (\ref{x2.23}) serves as a reduction
of (\ref{x232}). Thus it suffices to solve (\ref{x2.23}).

For a nontrivial solution with $W\neq0$, we may again assume $W>0$ and set $u=\ln W^2$ in (\ref{x2.23}) to arrive at
\be\label{x2.24}
u''=-4e^2 \e^u-2m^2_W,\quad -\infty<x<\infty.
\ee
This equation is similar to (\ref{x2.7}) and may be analyzed and  integrated as for the Abelian Higgs case considered earlier. Here we omit the discussion to avoid
redundancy.

Motivated by (\ref{x2.7}) and (\ref{x2.24}), we consider the integrability of the general equation
\be\label{x2.25}
u''=\lm (\e^u-\vep),
\ee
where $\lm,\vep$ are arbitrary constants. To proceed, setting $f=\e^u$, we have
\be\label{x2.26}
(\ln f)''=\lm (f-\vep).
\ee
Following the method of Liouville \cite{L,Ybook}, we further set $f=g'$ in (\ref{x2.26}) and integrate to get
\be\label{x2.27}
f'=\lm gg'-\lm\vep x g'+C_1 g'
\ee
where $C_1$ is an integration constant. Integrating (\ref{x2.27}) again, we obtain
\be\label{x2.28}
g'=\frac\lm2 g^2 -\lm\vep \left(xg-\int g(x)\,\dd x\right)+C_1 g +C_2,
\ee
where $C_2$ is also an integration constant. This equation is an integro-differential equation of the Riccati type and integrable when and only when $\vep=0$,
which will be studied further in the next section.
Our observation compares favorably with a similar integrability study carried out for the classical planar Liouville type equations based on the Painleve tests \cite{Schiff}.

\section{Domain walls in Abelian Chern--Simons--Higgs theories}\label{s4}
%\hskip\parindent \baselineskip 0.2in
%\renewcommand{\theequation}{4.\arabic{equation}}%
\setcounter{equation}{0}

We begin by considering  the gauged Schr\"{o}dinger equation over the $(2+1)$-dimensional Minkowski spacetime, coupled with a Chern--Simons electromagnetism, governed by the action density
\be\label{x3.1}
{\cal L}=-\frac{\kappa}2\epsilon^{\mu\nu\alpha}A_\mu\pa_\nu A_\alpha+
\ii\overline{\psi} D_0\psi-\frac1{2m}|D_j\psi|^2+\frac g2|\psi|^4,
\ee
where $\epsilon^{\mu\nu\alpha}A_\mu\pa_\nu A_\alpha=\frac12\epsilon^{\mu\nu\alpha}A_\mu F_{\nu\alpha}$ 
is the Chern--Simons term, $\psi$  a complex scalar field, $A_\mu$ the gauge potential, $D_\mu=\pa_\mu-\ii A_\mu$, and $\kappa,m,g>0$ are coupling
constants. The static equations of motion of (\ref{x3.1}) are
\be\label{x3.2}
A_0\psi=-\frac1{2m}D_j^2\psi-g|\psi|^2\psi,\quad
F_{12}=\frac1\kappa|\psi|^2,\quad
\pa_j A_0=-\frac\ii{2m\kappa}\ep_{jk}(\psi\overline{D_k\psi}-\overline{\psi}D_k\psi),
\ee
which are complicated. In \cite{JP1,JP2}, it is shown by Jackiw and Pi that, when $g=\frac1{m\kappa}$, (\ref{x3.2}) may be reduced into the BPS equations
\be\label{x3.3}
D_1\psi+\ii D_2\psi=0,\quad
F_{12}=\frac1{\kappa}|\psi|^2,\quad
A_0=-\frac1{2m\kappa}|\psi|^2,
\ee
among which the third equation is a constraint. We now pursue a domain wall substructure of (\ref{x3.2}) and (\ref{x3.3}) with the ansatz
\be\label{x3.4}
A_0=A(x),\quad A_1=0,\quad A_2=A(x),\quad\psi=\psi(x)=\mbox{real}.
\ee
Hence (\ref{x3.2}) and (\ref{x3.3}) become
\be\label{x3.5a}
A_0\psi=-\frac1{2m}\psi''+\frac1{2m}A^2\psi-g\psi^3,\quad A'=\frac1{\kappa}\psi^2,\quad A_0'=\frac1{m\kappa}\psi^2 A,\quad -\infty<x<\infty,
\ee
and
\be\label{x3.5b}
\psi'+A\psi=0,\quad A'=\frac1{\kappa}\psi^2,\quad A_0=-\frac1{2m\kappa}\psi^2,\quad -\infty<x<\infty,
\ee
respectively, in which the third equations are constraints which define $A_0$. With $g=\frac1{m\kappa}$, it may be examined that (\ref{x3.5b})
implies (\ref{x3.5a}). Namely, (\ref{x3.5a}) is reduced into (\ref{x3.5b}). Thus, with $u=\ln \psi^2$ in (\ref{x3.5b}), we obtain the further reduced equation
\be\label{x3.6}
u''=-\frac2\kappa \e^u,\quad -\infty<x<\infty,
\ee
which is contained as a special case of (\ref{x2.25}) when setting $\vep=0$, which is the one-dimensional version of the classical planar Liouville equation \cite{L}, and integrable, as commented in the previous section. In fact, (\ref{x3.6}) indicates that its solution is globally concave down and symmetric about its unique
global maximum, say $u_0$. Assume $u(x_0)=u_0$. Then, using $u'(x_0)=0$, we can integrate (\ref{x3.6}) to obtain
\be\label{x3.7}
\mbox{arctanh} \sqrt{1-\e^{u-u_0}}=\frac{\e^{\frac{u_0}2}}{\sqrt{\kappa}}(x-x_0),\quad x>x_0,
\ee
or explicitly,
\be\label{x3.8}
u(x)=u_0-2\ln\left(\cosh\left[\frac{\e^{\frac{u_0}2}}{\sqrt{\kappa}}(x-x_0)\right]\right),\quad -\infty<x<\infty,
\ee
which is even about $x=x_0$. Thus the general solution of the one-dimensional Liouville equation (\ref{x3.6}) depends on two arbitrary parameters, $x_0$ and 
$u_0$, and its asymptotes at $x=\pm\infty$ are determined by its global maximum $u_0$.

For completeness, we also consider the situation where $\kappa<0$ in (\ref{x3.6}). Then the solution is concave up and symmetric about its global minimum $u_0$
at $x_0$. Thus, as above, we obtain in a similar way the explicit solution
\be\label{x3.9}
u(x)=u_0+\ln\left(1+\tanh^2\left[\frac{\e^{\frac{u_0}2}}{\sqrt{-\kappa}}(x-x_0)\right]\right),\quad -\infty<x<\infty.
\ee
It is surprising to note that, unlike its planar version where all solutions blow up outside a compact region \cite{Sattinger}, the solutions of (\ref{x3.6}) for $\kappa<0$ are all globally defined over the full real line.

We now turn our attention to the relativistic Chern--Simons--Higgs theory introduced by Hong, Kim, and Pac \cite{HKP}, and Jackiw and Weinberg \cite{JW}.
In normalized
units and assuming the critical coupling, the Lagrangian action density of this relativistic Abelian
theory is written
\begin{equation}\label{x3.10}
{\cal L}=-\frac14{\kappa}\ep^{\mu\nu\alpha}A_\mu F_{\nu
\alpha}+D_\mu\phi\overline{D^\mu\phi}-\frac1{\kappa^2}|\phi|^2(1-|\phi|^2)^2,
\end{equation}
where  $\kappa\in \bfR$ is nonzero and $\phi$  a complex scalar field which can be viewed as 
a Higgs field.
The Euler--Lagrange equations of (\ref{x3.10}) are
\be
\frac12\kappa\ep^{\mu\nu\alpha}F_{\nu\alpha}=j^\mu,\quad
D_\mu D^\mu\phi=-\frac1{\kappa^2}(2|\phi|^2[|\phi|^2-1]+[|\phi|^2-1]^2)\phi,\label{x3.11}
\ee
where
$
 j^\mu=\ii(\phi\overline{D^\mu\phi}-\overline{\phi}D^\mu\phi),
\mu=0,1,2
$,
 is a conserved matter 
current density. These equations are rather complicated. In \cite{HKP,JW}, it is shown that (\ref{x3.11}) enjoys a BPS reduction as in the Abelian Higgs situation
given as the system of equations
\be\label{x3.12}
D_1\phi+\ii D_2\phi=0,\quad
F_{12}=\frac2{\kappa^2}|\phi|^2(1-|\phi|^2),\quad
\kappa F_{12}=2A_0|\phi|^2,
\ee
which may be derived from the energy density of a solution of (\ref{x3.11}), which reads
\be\label{x3.12b}
{\cal H}=\frac{\kappa^2 F_{12}^2}{4|\phi|^2}+|D_i\phi|^2+\frac1{\kappa^2}|\phi|^2(1-|\phi|^2)^2.
\ee
 It is again consistent to take the domain-wall ansatz (\ref{x3.4}) with $\psi$ being replaced by $\phi$ so that the systems (\ref{x3.11}) and (\ref{x3.12}) become
\be\label{x3.14}
\phi''-A^2\phi+A_0^2\phi=\frac1{\kappa^2}(\phi^2-1)(3\phi^2-1)\phi,\quad\kappa A_0'=2A\phi^2,\quad \kappa A'=2A_0\phi^2,
\ee
and
\be\label{6.3}
\phi'+A\phi=0,\quad
A'=\frac{2}{\kappa^2}\phi^2(1-\phi^2),\quad \kappa A'=2A_0\phi^2,
\ee
respectively, in which the last equations serve as  constraints. It may be checked that (\ref{6.3}) implies (\ref{x3.14}). That is, (\ref{6.3}) is a reduction of (\ref{x3.14}).
Again $\phi$ may be assumed to stay positive and the substitution $u=2\ln \phi$ recasts (\ref{6.3}) into
\be\label{6.4}
u''=\lm \e^u (\e^u-1),\quad -\infty<x<\infty,
\ee
with $\lm=\frac4{\kappa^2}$. Multiplying (\ref{6.4}) with $u'$ and integrating, we find
\be\label{6.5}
(u')^2(x)=(u')^2(x_0)+\lm(\e^{u(x)}-1)^2-\lm(\e^{u(x_0)}-1)^2.
\ee

Firstly, we are interested in a solution satisfying $u(-\infty)=0$. Thus, letting $x_0\to-\infty$ and using $u'(-\infty)=0$ in (\ref{6.5}), we obtain
\be\label{x319}
u'(x)=-\sqrt{\lm} (1-\e^{u(x)}),
\ee
which may be integrated to give us the implicit solution
\be
\frac{\e^u}{1-\e^u}=\frac{\e^{u_0}}{1-\e^{u_0}}\e^{-\sqrt{\lm}(x-x_0)},\quad u_0=u(x_0),\quad  x\in\bfR.
\ee
Or explicitly, we obtain the exact solution
\be\label{x321}
u(x)=\ln\left(\frac{\e^{u_0-\sqrt{\lm}(x-x_0)}}{1-\e^{u_0}+\e^{u_0-\sqrt{\lm}(x-x_0)}}\right),\quad x\in\bfR.
\ee
This solution which depends on two arbitrary parameters, $x_0$ and $u_0$, automatically satisfies the boundary condition $u(-\infty)=0, u(\infty)=-\infty$.

Thus, using the relation 
\be\label{x322}
\phi=\e^{\frac12 u},
\ee
$\phi_0=\phi(x_0)=\e^{\frac12 u_0}$, and (\ref{x321}), we obtain
\be\label{x323}
\phi(x)= \frac{\phi_0\e^{-\frac{\sqrt{\lm}}2(x-x_0)}}{\left(1-\phi_0^2+\phi_0^2\e^{-\sqrt{\lm}(x-x_0)}\right)^{\frac12}},\quad x\in\bfR,
\ee
which satisfies the boundary condition $\phi(-\infty)=1, \phi(\infty)=0$, linking the superconducting phase, or the Higgs phase as a result of
the spontaneously broken symmetry, with the normal phase, or the magnetic phase. Note that, at a first sight of (\ref{x323}), the solution depends on two
parameters $x_0\in\bfR$ and $\phi_0\in(0,1)$. However, due to translation invariance, the parameter $x_0$ is actually fixed by $\phi_0$. Thus the
full solution
family as given in (\ref{x323}) contains exactly one free parameter, $\phi_0\in(0,1)$.

Besides, for the domain-wall solution, the Hamiltonian density (\ref{x3.12b}) becomes
\be
{\cal H}=\frac{\kappa^2 (A')^2}{4\phi^2}+(\phi')^2+A^2\phi^2+\frac1{\kappa^2}\phi^2(1-\phi^2)^2,
\ee
which in view of the BPS equations (\ref{6.3}) may be reduced into the much simplified form
\be\label{x325}
{\cal H}=2(\phi')^2+\frac2{\kappa^2}\phi^2(1-\phi^2)^2.
\ee
With this and (\ref{x323}), we can directly evaluate the total energy of a domain wall soliton. Here, however, we may do so indirectly and straightforwardly as follows.

In fact, using (\ref{x325}), we have
\be\label{x326}
E=\int_{\bfR^2}{\cal H}\,\dd x=2\int_{\bfR}\left(\phi'+\frac1{\kappa}\phi(1-\phi^2)\right)^2\,\dd x+\frac1{\kappa}\int_{\bfR}((1-\phi^2)^2)'\,\dd x.
\ee
In view of (\ref{x319}) and (\ref{x322}), the first integral on the right-hand side of (\ref{x326}) vanishes; in view of the boundary condition on $\phi$, the second
integral on the right-hand side of (\ref{x326}) yields the value $\frac1{\kappa}$. Thus, we obtain the total energy of a domain-wall soliton connecting the 
superconducting or the Higgs  and the normal or the magnetic phases to be
\be\label{x327}
E=\frac1{\kappa},
\ee
which is independent of the parameters $x_0$ and $\phi_0$ of the solution expressed in (\ref{x323}).

In Figure \ref{F1}, we present the plots of $\phi, \phi'$, and the energy density $\cal H$ against the $x$ axis, with $\kappa=\frac13$, $x_0=0$, and $\phi_0=\frac12$. It is seen
that the phase transition is realized rapidly in a highly local region, $-1.2<x<1.2$, say. In fact, energetically, we have
\be
\int_{-1.2}^{1.2}{\cal H}\,\dd x=2.998492403,
\ee
which is amazingly close to the exact value given by (\ref{x327}), namely, $E=3$. Such a domain wall soliton will become more and more localized as $\kappa$ assumes smaller and smaller values. Here we omit the examples along this line.

\begin{figure}
\begin{center}
\includegraphics[height=6cm,width=7cm]{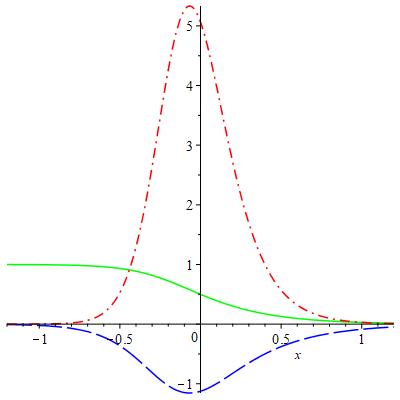}
\caption{Domain wall describing a phase transition between the superconducting and normal states represented by the boundary condition $\phi(-\infty)=1$
and $\phi(\infty)=0$, with $\kappa=\frac13$, $x_0=0$, and $\phi_0=\frac12$. The graphs of $\phi, \phi'$, and the energy density $\cal H$ are given by the solid, dash, and dash-dot  curves, respectively. It is
seen that $\cal H$ peaks at the spot where $\phi$ decreases most rapidly. }
\label{F1}
\end{center}
\end{figure}

Secondly, we aim to obtain a solution of (\ref{6.4}) satisfying the boundary condition $u(\pm\infty)=-\infty$. For such a solution, let $x_0$ be such that $u$ attains
its global maximum, say $u_0$, over $(-\infty,\infty)$. Then the maximum principle applied to (\ref{6.4}) implies $u_0<0$. Inserting $u'(x_0)=0$ into (\ref{6.5})
and integrating the resulting equation
\be\label{6.11}
u'=-\sqrt{\lambda}\sqrt{\e^{2u}-2\e^u+b},\quad x>x_0,
\ee
where $b=\e^{u_0}(2-\e^{u_0})>0$, we obtain
\be\label{6.12}
\int_{u_0}^{u(x)}\frac{\dd u}{\sqrt{\e^{2u}-2\e^u+b}}=-\sqrt{\lambda}(x-x_0).
\ee
Noting that $u(x)<u_0$ and rewriting the left-hand side of \eqref{6.12} as $I$, we have
\ber
I&=&-\int_{u_0}^{u(x)}\frac{\dd \e^{-u}}{\sqrt{b\e^{-2u}-2\e^{-u}+1}}\notag\\
&=&-\int_{v_0}^{v(x)}\frac{\dd v}{\sqrt{bv^2-2v+1}}\quad (v_0=\e^{-u_0},v=\e^{-u}, v(x)=\e^{-u(x)})\notag\\
&=&-\frac{1}{\sqrt{b}}\int_{v_0}^{v(x)}\frac{\dd v}{\sqrt{(v-v_1)(v-v_2)}}\quad \left(v_1=v_0=\e^{-u_0}, v_2=\frac1{2-\e^{u_0}}\right)\notag\\
&=&-\frac{1}{\sqrt{b}}\int_{v_0}^{v(x)}\frac{\dd v}{(v-v_1)\sqrt{\frac{v-v_2}{v-v_1}}},\label{6.13}
\eer
where we note that $v(x)>v_1>v_2$.
Now set $\sqrt{\frac{v-v_2}{v-v_1}}=w$. Then $w>1$ and
$v=\frac{v_2-v_1w^2}{1-w^2}$. Thus, with $\dd v=\frac{2(v_2-v_1)w}{(1-w^2)^2}\dd w$ and $v-v_1=\frac{v_2-v_1}{1-w^2}$,  we have
\ber
\int\frac{\dd v}{(v-v_1)\sqrt{\frac{v-v_2}{v-v_1}}}&=&2\int\frac{\dd w}{1-w^2}=\ln\frac{w+1}{w-1}\notag\\
&=&2\ln(\sqrt{v-v_1}+\sqrt{v-v_2})-\ln(v_1-v_2),\label{6.14}
\eer
Inserting \eqref{6.14} into \eqref{6.11}, we get
\be\label{6.16}
\sqrt{v-v_1}+\sqrt{v-v_2}=\sqrt{v_1-v_2}\,\e^{\frac{1}{2}\sqrt{\lambda b}\,(x-x_0)},\quad x>x_0,
\ee
which also gives  us
\be\label{6.15}
\sqrt{v-v_1}-\sqrt{v-v_2}=-\sqrt{v_1-v_2}\,\e^{-\frac12\sqrt{\lm b}\,(x-x_0)},\quad x>x_0.
\ee
From these equations, we obtain
\ber
v&=&\frac{(v_1+v_2)}2+\frac{(v_0-v_2)}2\cosh(\sqrt{\lambda b}\,(x-x_0))\notag\\
&=&\frac{1}{\e^{u_0}(2-\e^{u_0})}\left(1+(1-\e^{u_0})\cosh(\sqrt{\lambda b}\,(x-x_0))\right).\label{6.17}
\eer
Hence, returning to the original variable $u=-\ln v$, we arrive at the following expression for the solution
\be\label{6.18}
u(x)=u_0+\ln(2-\e^{u_0})-\ln\Big(1+(1-\e^{u_0})\cosh\big(\sqrt{\lambda \e^{u_0}(2-\e^{u_0})}\,(x-x_0)\big)\Big),\quad x>x_0.
\ee
The solution defined below $x_0$ is uniquely obtained from (\ref{6.18}) by an even-function extension, of course, due to the structure of the equation (\ref{6.4}).
However, since the function given in (\ref{6.18}) is automatically even about $x_0$, we see that it actually defines $u(x)$ for all $x\in\bfR$.

In summary, we have explicitly obtained the complete family of solutions of the Chern--Simons domain wall equation (\ref{6.4}) subject to the boundary condition $u(\pm\infty)=-\infty$. These solutions are uniquely determined by two arbitrarily prescribed parameters, $-\infty<x_0<\infty$ and $u_0<0$, where $x_0$
is the point where $u$ attains its global maximum $u_0$, about which the solution is an even function, which is given by the expression (\ref{6.18}) for $x>x_0$.
In particular, the solution enjoys the sharp asymptotic behavior
\be\label{6.19}
u(x)=\mp\sqrt{\lambda \e^{u_0}(2-\e^{u_0})}\,x+\mbox{O}(1),\quad x\rightarrow\pm\infty.
\ee
Since $\e^{u_0}(2-\e^{u_0})\to0$ as $u_0\to-\infty$ and $\e^{u_0}(2-\e^{u_0})\to 1$ as $u_0\to0$, we see that, up to a shift, the solution is uniquely determined by its asymptotes of the form
\be\label{6.20}
u(x)=\mp \sqrt{\lm} \vep \,x+\mbox{O}(1),\quad x\to\pm\infty,
\ee
where $\vep$ may assume any value in the unit interval $(0,1)$. Note also that, since $\e^{u_0}(2-\e^{u_0})$ increases for $u_0<0$, we see that a solution
of higher maximum decays faster asymptotically, which is interesting. In fact, explicitly, for any $\vep\in(0,1)$
given in (\ref{6.20}), we may solve $\e^{u_0}(2-\e^{u_0})=\vep^2$ to get
the prescribed maximum $u_0<0$ of $u(x)$ to be
\be
u_0=\ln\left(1-\sqrt{1-\vep^2}\right).
\ee

From (\ref{x322}) and (\ref{6.18}), we may return to the original variable to get
\be\label{x3.39}
\phi(x)=\frac{\phi_0 (2-\phi_0)^{\frac12}}{\left(1+(1-\phi_0)\cosh\sqrt{\lm \phi_0(2-\phi_0)}(x-x_0)\right)^{\frac12}},\quad x\in\bfR,\quad \phi(x_0)=\phi_0\in (0,1).
\ee
Inserting this into (\ref{x325}), we may evaluate the total energy $E=\int_{\bfR}{\cal H}\,\dd x$ to obtain a closed-form expression $E=E(\phi_0)$, independent
of $x_0$ but dependent on $\phi_0$, which is always finite but too complicated to present here. Nevertheless, it is still possible to reduce the amount of computation
if we use (\ref{x326}) instead of  evaluating (\ref{x325}) directly. In so doing, we may insert (\ref{x3.39}) into (\ref{x326}),  with $x_0=0$ by translation
invariance, and rewrite (\ref{x326}) as
\be\label{x3.40}
E(\phi_0)=4\int_0^\infty \left(\phi'+\frac1{\kappa}\phi(1-\phi^2)\right)^2\,\dd x+\frac2{\kappa}\left(1-(1-\phi^2_0)^2\right).
\ee
The integral on the right-hand side of (\ref{x3.40}) does not vanish but can be integrated to render an elementary function of $\phi_0$, which we omit.

In Fugure \ref{F2}, we present the profiles of a solution with $\kappa=1$ and $\phi_0=\frac12$. Since $\phi(\pm\infty)=0$, the solution represents
an instanton-like lump \cite{R} rather than a domain wall whose quantum-mechanical meaning is of separate interest
(specifically, the instanton interpretation comes up handy when the spatial coordinate $x$ is regarded to be an imaginarized time coordinate via a Wick rotation \cite{Wick}). The solution is localized. In fact,
the integral part of (\ref{x3.40}) has the value $1.180053251$ and the same evaluated over the truncated interval $(0,6)$ has the value $1.179867364$.
Furthermore, it may be seen that
more localized profiles result from smaller values of $\kappa$ as before. We omit the examples.

Unlike the domain wall soliton linking the superconducting and normal phases realized by the boundary condition $\phi(-\infty)=1, \phi(\infty)=0$, in the present
instanton-like situation, $\phi(\pm)=0$, the total energy of the solution depends on $\phi_0\in(0,1)$. In Figure \ref{F3}, we display the plots of $E(\phi_0)$ for $\kappa=1,2,4$ against $\phi_0\in(0,1)$. We see that $E(\phi_0)$ decreases when $\kappa$ increases and increases when $\phi_0$ increases. These results are
in sharp contrast against the energy identity (\ref{x327}) established for the domain wall solitons in the previous situation.

\begin{figure}
\begin{center}
\includegraphics[height=6cm,width=7cm]{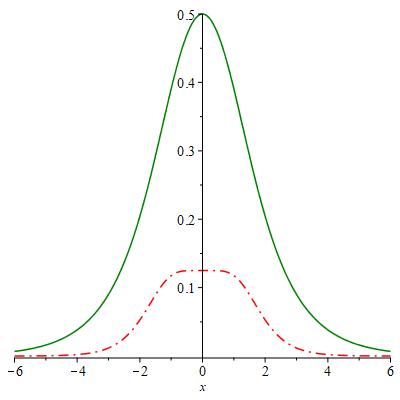}
\caption{An instanton-like solution of the domain-wall equation subject to the boundary condition $\phi(\pm\infty)=0$ with $\kappa=1$ and $\phi_0=\frac12$. The graphs of $\phi$ and its energy
density $\cal H$ are depicted by solid and dash-dot curves, respectively. Both $\phi$ and $\cal H$ are highly localized.}
\label{F2}
\end{center}
\end{figure}

\begin{figure}
\begin{center}
\includegraphics[height=7cm,width=7cm]{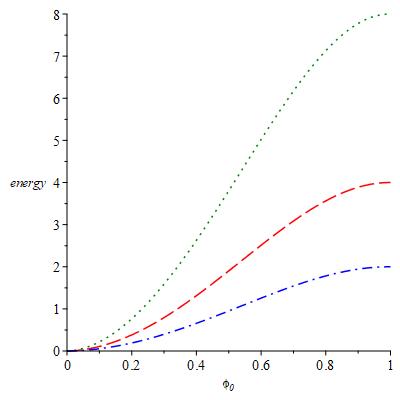}
\caption{Dependence of the energy of the solution on $\kappa$ and $\phi_0$. The dot, dash, and dash-dot curves correspond to $\kappa=1,2,4$,
respectively. It is seen that the energy is an increasing function of $\phi_0$ and greater values of $\kappa$ give rise to smaller values of the energy. }
\label{F3}
\end{center}
\end{figure}

In conclusion, we summarize our study of this section on the domain wall solitons arising in nonrelativistic and relativistic Abelian Chern--Simons--Higgs
theories defined by (\ref{x3.1}) and (\ref{x3.10}), respectively, as follows.

\begin{theorem}
The solutions of the BPS domain-wall equations (\ref{x3.5b}) and (\ref{6.3}) are governed by the one-dimensional Liouville type equations
(\ref{x3.6}) and (\ref{6.4}), respectively, which are both integrable.

\begin{enumerate}
\item[(i)] For any $\kappa\neq0$, the full sets of solutions of (\ref{x3.6}) are given by the explicit formulas (\ref{x3.8}) and (\ref{x3.9}), for $\kappa>0$
and $\kappa<0$, respectively, where the two parameters $x_0$ and $u_0$ are arbitrary. In particular, unlike the classical planar Liouville equation whose
 solutions when $\kappa<0$ break down outside some local regions, the one-dimensional Liouville equation does not suffer from such a breakdown property
and all solutions are globally defined.

\item[(ii)] For any $\lm>0$, the full set of solutions of (\ref{6.4}) satisfying the boundary condition $u(-\infty)=0$ and $u(\infty)=-\infty$ is given by (\ref{x321})
where $x_0\in\bfR$ and $u_0<0$. The total energy of such a solution is independent of the parameters $x_0,u_0$ but only depends on the coupling
parameter $\lm$. The parameters $x_0$ and $u_0$ are not free parameters but determine each other.

\item[(iii)] For any $\lm>0$, the solutions of (\ref{6.4}) satisfying the boundary condition $u(\pm\infty)=-\infty$ are all even functions about
$x=x_0\in\bfR$ given by (\ref{6.18}) for $x>x_0$ so that each of them achieves its prescribed global maximum $u_0\in (-\infty,0)$ at $x_0$. The total energy of such a
solution depends on its global maximum $u_0$ as well but is independent of $x_0$ due to translation invariance. The parameters $x_0$ and $u_0$ are free parameters.
\end{enumerate}
\end{theorem}

Of course, the explicit knowledge regarding the domain wall solitons arising in the Abelian gauge field theory acquired in the last and present sections is useful for us to obtain 
exact information on local and global physical quantities in the models. For example, for the Abelian Chern--Simons--Higgs theory defined by the Lagrangian density
(\ref{x3.10}), the magnetic field is $H=F_{12}$ and the electric charge density resulting from the Gauss law constraint is $\rho=2A_0|\phi|^2=\kappa F_{12}$ given
as the third equation in (\ref{x3.12}). Within the domain wall framework, these fields are reduced into
\be
H=A'=\frac2{\kappa^2}\phi^2(1-\phi^2),\quad \rho=\kappa A'=\frac2\kappa \phi^2(1-\phi^2),
\ee
as given in (\ref{6.3}). Thus, applying (\ref{x323}), both $H$ and $\rho$ are explicitly obtained for the domain wall soliton realizing the phase transition from the
Higgs domain (at $x=-\infty$) to the magnetic domain (at $x=\infty$). In view of $u=2\ln\phi$, $\lm=\frac4{\kappa^2}$, (\ref{6.4}), and (\ref{x319}), we see that the total magnetic charge is
\be
Q_m=
\int_{-\infty}^\infty H\,\dd x=-\frac\lm2 \int_{-\infty}^\infty \e^u(\e^u-1)\,\dd x=-\frac12 (u'(\infty)-u'(-\infty))=\frac1\kappa,
\ee
and the total electric charge  is also determined to be $Q_e=\int_{-\infty}^\infty \rho\,\dd x=\kappa Q_m=1$.

Moreover, recall that for the Abelian Chern--Simons--Higgs BPS vortex equations, (\ref{x3.12}), over $\bfR^2$, the associated nonlinear elliptic equation describing
a distribution of vortices at the prescribed points $p_1,\dots,p_N$ is \cite{HKP,JW}
\be\label{x342}
\Delta u=\lm\e^u(\e^u-1)+4\pi\sum_{s=1}^N\delta_{p_s}(x),
\ee
which is known to be more difficult than its Abelian Higgs counterpart, (\ref{x2.4}), and also non-integrable \cite{Schiff}. The solutions satisfying $u=0$ and
$u=-\infty$ at infinity are referred to as topological and non-topological, respectively \cite{JLW,JP2,JPW}. Although existence of both kinds of solutions has
been established for a long time \cite{Chae,Chan,CHMcY,SYcs1,SYcs2}, a full description of such solutions is still elusive. Hopefully, our comprehensive and
complete knowledge about the solutions of the one-dimensional version of the equation, (\ref{6.4}), as summarized in (ii) and (iii) of the above theorem, offers useful new insight
into the solutions of (\ref{x342}).

\section{Domain walls in non-Abelian gauge field theory}
\setcounter{equation}{0}

We now follow Bolognesi {\em et al} \cite{Bol} to consider the simplest non-Abelian $U(2)$-gauge field 
\be
A_\mu=\frac{a_\mu}2 {\bf 1}+\frac{ A_\mu^a}2 \sigma^a
\ee
interacting with a Higgs field $q$ given as a $2\times2$ complex matrix for which the gauge-covariant derivative is defined by
\be
D_\mu =\pa_\mu-\ii \frac{e a_\mu}2{\bf1}-\ii \frac{g A_\mu^a}2\sigma^a,\quad e,g>0.
\ee
Since the Pauli matrices $\sigma^a$ ($a=1,2,3$) obey the commutator relation $\left[{\sigma^a},{\sigma^b}\right]=2\ii \epsilon^{abc}{\sigma^c}$, the field strength tenors or the curvatures may be deduced from the commutator
\be
(D_\mu D_\nu-D_\nu D_\mu)q=-\ii\frac{e f_{\mu\nu}}2 q -\ii\frac{g F^a_{\mu\nu}}2 \sigma^a q,
\ee
where
\be
f_{\mu\nu}=\pa_\mu a_\nu-\pa_\nu a_\mu,\quad F^a_{\mu\nu}=\pa_\mu A^a_\nu -\pa_\nu A^a_\mu-\ii\frac{g}2 \epsilon^{abc}A^b_\mu A^c_\nu.
\ee
The coupled $U(2)$ gauge field and Higgs particle theory is defined by the BPS Lagrangian action density \cite{Bol}
\ber\label{2.1}
\mathcal{L}=-\frac{1}{4}f_{\mu \nu}f^{\mu \nu}-\frac{1}{4}F_{\mu \nu}^aF^{\mu \nu a}+\text{Tr}(D_\mu q)^\dag(D^\mu q)-\frac{e^2}{8}(|q|^2-2\xi)^2-\frac{g^2}{8}\sum_a\text{Tr}(q^\dag\sigma^aq)^2,
\eer
where $|q|^2=\text{Tr}(qq^\dag)$. The equations of motion of (\ref{2.1}) are complicated but in static two dimensions they possess the BPS reduction \cite{Bol}
\be\label{2.2}
f_{12}+\frac{e}{2}(|q|^2-2\xi)=0,\quad
F_{12}^a+\frac{g}{2}\text{Tr}(q^\dag\sigma^aq)=0,\quad
D_1q+\ii D_2 q=0,
\ee
as in the Abelian Higgs theory \cite{B,JT}. Moreover, as in \cite{Bol}, the system (\ref{2.2}) is seen to possess
a domain wall structure interpolating the Higgs and the magnetic phases given by the ansatz
\ber\label{2.5}
&&a_1=0,\quad a_2=a(x),\quad A_i^1=A_i^2=0,\quad i=1,2,\nn\\
&& A_1^2=0,\quad  A_2^3=A(x),\quad
q=\mbox{diag}\{q_1(x),q_2(x)\}=\mbox{real},
\eer
(note that here and in the sequel the function $a(x)$ should not be confused with the index $a$ used to  label the Pauli matrices) which renders (\ref{2.2}) into
\ber\label{2.6}
&& a'+\frac{e}{2}(q_1^2+q_2^2-2\xi)=0,\quad
A'+\frac{g}{2}(q_1^2-q_2^2)=0,\nn\\
&&q_1'+\left(\frac{e}{2}a+\frac{g}{2}A\right)q_1=0,\quad
q_2'+\left(\frac{e}{2}a-\frac{g}{2}A\right)q_2=0.
\eer
Thus, with $\gamma=\frac{g^2}{e^2}$, the above system is reduced into the following coupled second-order equations in terms of $q_1,q_2>0$:
\ber 
(\ln q_1)''&=&\frac{e^2}{4}\left((1+\gamma)q_1^2+(1-\gamma)q_2^2-2\xi\right),\label{2.10}\\
(\ln q_2)''&=&\frac{e^2}{4}\left((1-\gamma)q_1^2+(1+\gamma)q_2^2-2\xi\right).\label{2.11}
\eer
Furthermore, setting $u_i= \ln q_i^2-\ln\xi, i=1, 2$, and $\lm=\frac12 e^2\xi$, then the equations \eqref{2.10}--\eqref{2.11} become
\ber 
u_1''&=&\lm\left((1+\gamma)\text e^{u_1}+(1-\gamma)\text e^{u_2}-2\right),\label{2.14}\\
u_2''&=&\lm\left((1-\gamma)\text e^{u_1}+(1+\gamma)\text e^{u_2}-2\right).\label{2.15}
\eer
The following boundary conditions are of interest:
\ber
&& u_1(\pm\infty)=u_2(\pm\infty)=0\quad \mbox{(Higgs to Higgs phase)},\label{4.16a}\\
&& u_1(-\infty)=u_2(-\infty)=0,\quad u_1(\infty)=u_2(\infty)=-\infty\quad \mbox{(Higgs to magnetic phase)},\quad\quad \label{4.16b}\\
&&u_1(\pm\infty)=u_2(\pm\infty)=-\infty\quad\mbox{(magnetic to magnetic phase)}.\label{4.16c}
\eer

We first consider (\ref{2.14})--(\ref{2.15}) subject to the boundary condition (\ref{4.16a}) and show that there is no nontrivial solution. In fact, setting
$
U_1=\frac{1}{2}(u_1+u_2),U_2=\frac{1}{2}(u_1-u_2),
$
Then the equations \eqref{2.14}--\eqref{2.15} become
\ber\label{2.17}
U_1''=\lm(\e^{U_1+U_2}+\e^{U_1-U_2}-2),\quad U_2''=\lm\gamma\e^{U_1}(\e^{U_2}-\e^{-U_2}).
\eer
Note that the second equation in \eqref{2.17} gives us
$
U_2''=\lm\gamma\text e^{U_1}(\text e^V+\text e^{-V})U_2
$
where $V$ lies between $0$ and $U_2$. Thus we have $U_2= 0$ in view of (\ref{4.16a}) and the maximum principle. Inserting this result into the first equation
in (\ref{2.17}) and using the same argument, we deduce $U_1= 0$. Consequently, the only solution of (\ref{2.14})--(\ref{2.15}) subject to the boundary condition
(\ref{4.16a}) is the zero solution.

We next consider the boundary conditions (\ref{4.16b}) and (\ref{4.16c}). The structure of the equations (\ref{2.14})--(\ref{2.15}) indicates that it is consistent to take the ansatz
$u_1=u_2=u$ which reduces the equations into a single one, $u''=2\lm (\e^u-1)$, which is well studied in Section 2. There we have seen that the solutions
satisfying the boundary conditions $u(-\infty)=0, u(\infty)=-\infty$ and $u(\pm\infty)=-\infty$, respectively, may all be constructed and described. In other words, we have 
established the existence of  solutions of (\ref{2.14})--(\ref{2.15}) satisfying the boundary conditions (\ref{4.16b}) and (\ref{4.16c}) through using a single equation. Naturally this raises a question whether all solutions of (\ref{2.14})--(\ref{2.15}) are given by the ansatz $u_1=u_2$ since these equations are invariant
when one interchanges $u_1$ and $u_2$. In \cite{Bol}, a numerical example is presented showing that the equations (\ref{2.14})--(\ref{2.15}) possess a solution
with $u_1\neq u_2$. In fact, it is such that $u_1(-\infty)=u_2(-\infty)=0$ (realizing the Higgs phase) and $u_1(\infty)=-\infty,u_2(\infty)\equiv\overline{u}_2,-\infty<\overline{u}_2<\infty$ (realizing a mixed magnetic and Higgs phase). It is easy to see that, if such a solution exists, then consistency in (\ref{2.15}) implies
$\overline{u}_2=\ln\left( \frac2{1+\gamma}\right)$. In this part of the work, we shall prove that the equations (\ref{2.14})--(\ref{2.15}) do have solutions with
$u_1\neq u_2$. In order to do so, we shall consider the boundary condition (\ref{4.16c}). Regarding this, our existence result, phrased in terms of
the scalar fields $q_1,q_2$ as a solution to the original governing equations (\ref{2.10}) and (\ref{2.11}), is stated as follows.

\begin{theorem}\label{th4.1}
For any coupling parameters $e,\gamma,\xi>0$ in the domain wall equations (\ref{2.10}) and (\ref{2.11}) arising in the $U(2)$ non-Abelian Yang--Mills--Higgs theory defined by the action density
(\ref{2.1}), and any prescribed parameters $\alpha_1,\alpha_2,\beta_1,\beta_2$ satisfying the condition
\be\label{417}
(1+\gamma)(\alpha_1+\beta_1)+(\gamma-1)(\alpha_2+\beta_2)>0,\quad (\gamma-1)(\alpha_1+\beta_1)+(1+\gamma)(\alpha_2+\beta_2)>0,
\ee
there is a solution $(q_1,q_2)$ fulfilling the boundary condition $q_1(\pm\infty)=0,q_2(\pm\infty)=0$ with the sharp decay estimates
\bea
q^2_i(x)&=&\mbox{\rm O}(\e^{\alpha_i x-\frac12 e^2\xi x^2})\quad\mbox{as }x\to\infty,\quad i=1,2,\\
q_i^2(x)&=&\mbox{\rm O}(\e^{-\beta_i x-\frac12 e^2\xi x^2})\quad\mbox{as }x\to-\infty,\quad i=1,2.
\eea
Moreover, there hold the following exact results
\bea
\int_{\bfR} q_1^2(x)\,\dd x&=&\frac{1}{2e^2\gamma}\left((1+\gamma)(\alpha_1+\beta_1)+(\gamma-1)(\alpha_2+\beta_2)\right),\\
\int_{\bfR}q_2^2(x)\,\dd x&=&\frac{1}{2e^2\gamma}\left((\gamma-1)(\alpha_1+\beta_1)+(1+\gamma)(\alpha_2+\beta_2)\right).
\eea
For the given parameters, the above described solution is in fact unique, and the condition (\ref{417}) is also necessary for existence.
\end{theorem}

As in the Abelian gauge field theory situation, these solutions represent localized lumps and may be interpreted as instantons.

\section{Analysis and proof of existence by  direct minimization}\label{s5}
%\hskip\parindent \baselineskip 0.2in
%\renewcommand{\theequation}{5.\arabic{equation}}%
\setcounter{equation}{0}

To proceed, we use $u_{01}$ and $u_{02}$ to denote two background functions which are smooth and satisfy
\be\label{4.1}
u_{0i}=\alpha_i x,\quad x\geq1,\quad u_{0i}=-\beta_i x,\quad x\leq -1,\quad i=1,2,
\ee
where $\alpha_i,\beta_i$ ($i=1,2$) are suitable parameters to be determined later.
Set $\omega=\lm x^2$ and $u_i=\eta_i +u_{0i}-\omega, i=1, 2$. Then equations \eqref{2.14}--\eqref{2.15} become
\ber
\eta_1''&=&\lm\Big((1+\gamma)\text e^{\eta_1+u_{01}-\omega}+(1-\gamma)\text e^{\eta_2+u_{02}-\omega}\Big)-u_{01}'',\label{4.2a}\\
\eta_2''&=&\lm\Big((1-\gamma)\text e^{\eta_1+u_{01}-\omega}+(1+\gamma)\text e^{\eta_2+u_{02}-\omega}\Big)-u_{02}'',\label{4.2b}
\eer
or their matrix form
\be\label{4.3a}
\Gamma^{-1}(\eta_1'', \eta_2'')^\tau=\lm\left(\text e^{\eta_1+u_{01}-\omega}, \text e^{\eta_2+u_{02}-\omega}\right)^\tau-\Gamma^{-1}(u_{01}'',  u_{02}'')^\tau,
\ee
where $\tau$ denotes matrix transpose and
\be
\Gamma=
\left(
\begin{array}{cc}
1+\gamma & 1-\gamma\\
1-\gamma & 1+\gamma
\end{array}
\right),
\quad
\Gamma^{-1}=\frac{1}{4\gamma}
\left(
\begin{array}{cc}
1+\gamma & \gamma-1\\
\gamma-1 & 1+\gamma
\end{array}
\right),
\ee
so that $2\gamma$ and $2$ are the two eigenvalues of $\Gamma$.
Thus we see that (\ref{4.2a})--(\ref{4.2b}) or (\ref{4.3a}) may be derived as the Euler--Lagrange equations of the energy functional
\begin{equation}\label{4.4}
I(\eta_1,\eta_2)=\int\frac{1}{2}(\eta_1',\eta_2')\Gamma^{-1}(\eta_1',\eta_2')^\tau +\lm\int\left(\text e^{\eta_1+u_{01}-\omega} +\e^{\eta_2+u_{02}-\omega}\right)-\frac1{4\gamma}J(\eta_1,\eta_2),
\ee
in which we have set
\be\label{J}
J(\eta_1,\eta_2)=\int\left\{\Big((1+\gamma)u_{01}''+(\gamma-1)u_{02}''\Big)\eta_1+\Big((\gamma -1)u_{01}''+(1+\gamma)u_{02}''\Big)\eta_2\right\},
\ee
where and in the sequel we omit the domain of integration, $\bfR$, and the Lebesgue measure, $\dd x$, when there is no risk of confusion.
Of course, $u_{01}''$ and $u_{02}''$ are of compact supports and satisfy
\ber
\int\Big((1+\gamma)u_{01}''+(\gamma-1)u_{02}''\Big)&=&\Big((1+\gamma)(\alpha_1+\beta_1)+(\gamma-1)(\alpha_2+\beta_2)\Big)
\equiv2\gamma{\kappa_1},\label{4.5}\\
\int\Big((\gamma-1)u_{01}''+(1+\gamma)u_{02}''\Big)&=&\Big((\gamma-1)(\alpha_1+\beta_1)+(1+\gamma)(\alpha_2+\beta_2)\Big)
\equiv2\gamma{\kappa_2}.\label{4.5b}
\eer
For technical reasons, we will need to impose some conditions on the ranges of the parameters $\alpha_i,\beta_i$ ($i=1,2$) in due course later.

To proceed further, we need to consider a suitable weighted Sobolev space formalism.

Let  $h_0(x)\in C^{\infty}(\mathbb{R})$ be a positive-valued weight function satisfying $h_0(x)=\text e^{-\beta|x|},|x|\geq1$, where $\beta>0$ is a constant.
Set $\dd\mu=h_0(x)\,\dd x$ (where $\mu$ should not be confused with the spacetime index used elsewhere in this paper) and use $L^p(\dd x)$ and $L^p(\dd\mu)$
to denote the usual $L^p$-spaces over $\bfR$ with respect to the measures $\dd x$ and $\dd\mu$, respectively.
Use $\mathscr{H}$ to denote the Hilbert space obtained by taking the completion of $C^\infty_0(\bfR)$ under the norm given by
\be\label{xx5.10}
\|u\|^2_\mathscr{H}=\|u'\|^2_{L^2(\text{d}x)}+\|u\|^2_{L^2(\text{d}\mu)}.
\ee
Then $\mathscr{H}$ contains all constant functions and thus $u \mapsto\int_{\mathbb{R}} u\,\text{d}\mu$ is a continuous linear functional on $\mathscr{H}$.
Consequently,
\be
\dot{\mathscr{H}}=\left\{u\in \mathscr{H}\,\Big|\int u\,\text{d}\mu=0\right\}
\ee
is a closed subspace of $\mathscr{H}$. Therefore we have for each $u\in \mathscr{H}$ the decomposition
\ber
u=\bar{u}+\dot{u}, ~~\bar{u}\in\bfR, ~~\dot{u}\in\dot{\mathscr{H}}\label{4.6}.
\eer

 First, we establish the following Trudinger--Moser inequality \cite{Aubin,M,Trudinger} on $\mathbb{R}$, specializing on the method in McOwen \cite{Mc}.

\begin{lemma}\label{4.1} Let $\beta>0$ be the exponent in the weight function $h_0$ in the definition of the weighted measure $\dd\mu$.
For any $a\in{\mathbb{R}}$ and $b\in(0, \beta)$, there is some constant $C(b)>0$, so that
\ber\label{x5.13}
\int \exp(a|v|)\,\dd\mu\leq C(b)\exp\Big(\frac{a^2}{4b}\int (v'(x))^2\,\dd x\Big), \quad \forall v\in\mathscr{H},\quad \int v\,\dd\mu=0.
\eer
\end{lemma}
\begin{proof}
Let $u\in L^2(\text{d}\mu)$ be such that $u'\in L^2(\text{d}x)$ and
\ber \label{4.7}
\int (u'(x))^2\text{d}x=1,\quad \int u(x)\text{d}\mu=0.
\eer
Then there is a constant $K>0$ independent of $u$ but dependent on $b$ such that
\ber \label{4.8}
\int\text e^{b u^2}\text{d}\mu\leq K,\quad \forall b<\beta.
\eer

In fact, we have for $x>y$ the estimate
\ber\notag
|u(x)-u(y)|\leq\int_{y}^{x}|u'(t)|\text{d}t\leq\Big(\int_{y}^{x}\text{d}t\Big)^{\frac{1}{2}}\Big(\int_{y}^{x}(u'(t))^2\text{d}t\Big)^{\frac{1}{2}}.
\eer
Hence
\be \label{4.10}
(u(x)-u(y))^2\leq(x-y)\int_{y}^{x}(u'(t))^2\text{d}t\leq(x-y),
\ee
in view of \eqref{4.7}. Using \eqref{4.10} we have
\ber \label{4.11}
-|x-y|^{\frac{1}{2}}\leq u(x)-u(y)\leq|x-y|^{\frac{1}{2}}, ~~\forall~x, y\in\mathbb{R}.
\eer
Now multiplying \eqref{4.11} by $h_0(y)$, integrating against the measure $\text{d}y$ over $\mathbb{R}$, and using the second condition in \eqref{4.7}, we have
\ber \label{4.12}
|u(x)|\int_{\mathbb{R}}\text{d}\mu\leq\int_{\mathbb{R}}|x-y|^{\frac{1}{2}}h_0(y)\text{d}y.
\eer

On the other hand, we have
\ber \label{4.13}
|x-y|^{\frac{1}{2}}\leq(|x|+|y|)^{\frac{1}{2}}\leq|x|^{\frac{1}{2}}+|y|^{\frac{1}{2}},
\eer
inserting \eqref{4.13} into \eqref{4.12} we get
\ber \label{4.14}
|u(x)|\int_{\mathbb{R}}\text{d}\mu\leq|x|^{\frac{1}{2}}\int_{\mathbb{R}}\text{d}\mu+\int_{\mathbb{R}}|y|^{\frac{1}{2}}h_0(y)\text{d}y.
\eer
That is,
\ber \label{4.15}
|u(x)|\leq|x|^{\frac{1}{2}}+\frac{1}{\int_{\mathbb{R}}\text{d}\mu}\int_{\mathbb{R}}|y|^{\frac{1}{2}}h_0(y)\text{d}y\equiv|x|^{\frac{1}{2}}+C,\quad \forall x.
\eer
Thus,  using the interpolation inequality $2ab\leq\varepsilon a^2+\frac{1}{\varepsilon}b^2$, we arrive at
\ber \label{4.16}
u^2(x)\leq(1+\varepsilon)|x|+C^2\left(1+\frac{1}{\varepsilon}\right), \quad \forall x\in\mathbb{R}, \quad \forall\varepsilon>0.
\eer

Now choose $\varepsilon>0$ small enough such that $b(1+\varepsilon)<\beta$. In view of this condition and \eqref{4.16}, we see there is a constant $K>0$ such that
(\ref{4.8}) holds.

We are now ready to establish the Trudinger--Moser inequality (\ref{x5.13}).
For this purpose, assume $v\in L^2(\text{d}\mu)$, so that $v\not\equiv0$ but $\int v\,\text{d}\mu=0$. Then
$
\int (v'(x))^2\,\text{d}x>0.
$
For such a function $v$, we set
$u=v/\Big(\int (v'(x))^2\text{d}x\Big)^\frac{1}{2}\equiv \frac{v}{C_v}.
$
Then $u$ satisfies \eqref{4.7}.
Finally, for any constant $a\in{\mathbb{R}}$, using
\ber \label{4.21}
a|v|=a|C_vu|\leq b u^2+\frac{1}{4b}a^2C^2_v,~~\forall b>0
\eer
and \eqref{4.8}, we have
\ber \label{4.22}
\int_{\mathbb{R}}\exp(a|v|)\text{d}\mu\leq K(b)\exp(\frac{a^2}{4b}C^2_v),~~0<b<\beta.
\eer
So we get the asserted Trudinger--Moser  inequality (\ref{x5.13}).
\end{proof}

Next, we show that the Poincar\'{e} inequality holds in our context.

\begin{lemma}\label{L2}
There is a constant $C>0$ so that
\be\label{x5.25}
\|v\|^2_{L^2(\dd\mu)}\leq C\|v'\|^2_{L^2(\dd x)},\quad v\in\dot{\mathscr{H}}.
\ee
\end{lemma}
\begin{proof} For $v\in\dot{\mathscr{H}}$ and $\|v'\|_{L^2(\dd x)}\neq0$, set $u=v/{\|v'\|_{L^2(\dd x)}}$. Then (\ref{4.7}) holds and \eqref{4.16} gives us
\ber\label{4.23}
\|u\|^2_{L^2(\text{d}\mu)}=\int u^2 h_0(x)\text{d}x\leq\int (C_1|x|+C_2) h_0(x)\text{d}x\leq C,
\eer
independent of $u$ otherwise. In other words $\|v\|^2_{L^2(\text{d}\mu)}\leq C\|v'\|^2_{L^2(\text{d}x)}$ as stated.
\end{proof}

We then establish the following embedding property.

\begin{lemma}\label{4.3}
The injection  ${\mathscr{H}}\rightarrow L^2(\mathbb{R}, \text{d}\mu)$ is a compact embedding.
\end{lemma}
\begin{proof} Let $\{u_n\}$ be a weakly convergent sequence in $\HH$ and $u\in \HH$ its weak limit. Then the standard compact embedding 
$W^{1,2}(-R,R)\to C[-R,R]$ for any $R>0$  implies that $u_n\to u$ in $C[-R,R]$ for any $R>0$. In particular, $\{u_n(0)\}$ is a bounded sequence. Thus
\be
|u_n(x)|\leq|u_n(0)|+\left|\int_0^x u'_n(y)\,\dd y\right|\leq C_1+|x|^{\frac12}\|u_n'\|_{L^2(\dd x)}\leq C_1+C_2 |x|^{\frac12},\quad n=1,2,\dots,
\ee
where $C_1,C_2>0$ are constants. Thus, we see that, for any $\vep>0$, there is some $R_\vep>0$, such that
\be\label{x5.28}
\left(\int_{-\infty}^{-R_\vep}+\int_{R_\vep}^{\infty}\right) (u_n^2+u^2)\,\dd\mu<\vep,\quad n=1,2,\dots.
\ee
Therefore, in view of (\ref{x5.28}),  we arrive at
\bea
\limsup_{n\to\infty}\|u_n-u\|^2_{L^2(\dd\mu)}&=&\lim_{n\to\infty}\int_{-R_\vep}^{R_\vep}(u_n-u)^2\,\dd\mu+
\limsup_{n\to\infty}\left(\int_{-\infty}^{-R_\vep}+\int_{R_\vep}^{\infty}\right) (u_n-u)^2\,\dd\mu\nn\\
&\leq& 2\vep.
\eea
Since $\vep>0$ is arbitrary, we have proved that $\|u_n-u\|_{L^2(\dd\mu)}\to0$ as $n\to\infty$.
\end{proof}
\medskip

We are now prepared to prove the existence of a solution to the domain wall equations (\ref{4.2a})--(\ref{4.2b}) by a direct minimization of the functional $I$
given as in (\ref{4.4}) over the weighted Sobolev space $\HH$. That is, we are to solve the optimization problem
\ber \label{4.25}
\min\left\{I(\eta_1, \eta_2)\,\Big|\, \eta_1, \eta_2\in\mathscr{H}\right\}.
\eer
 In view of the lemmas above, the functional $I$ is $C^1$ and weakly lower semicontinuous over $\HH$. Here we omit the discussion. Below we
focus on the establishment of the
coerciveness of $I$ over $\HH$. 

First, with the notation given in (\ref{4.6}), we see that for the functional $J$ defined in (\ref{J}) we have
\be\label{x531}
|J(\dot{\eta}_1,\dot{\eta}_2)|\leq \vep\left(\|\dot{\eta}'_1\|_{L^2(\dd x)}^2+\|\dot{\eta}'_2\|^2_{L^2(\dd x)}\right)+C(\vep),
\ee
by applying the Schwarz inequality and Lemma \ref{L2}, where $\vep>0$ is arbitrary and $C(\vep)$ a constant depending on $\vep$.

Next, using (\ref{4.5}) and (\ref{4.5b}), we have
\be\label{x532}
J(\bar{\eta}_1,\bar{\eta}_2)=2\gamma\left(\kappa_1\bar{\eta}_1+\kappa_2\bar{\eta}_2\right).
\ee

Furthermore, by Jensen's inequality, we have
\bea\label{x533}
\int \e^{\eta_i+u_{0i}-\omega}\,\dd x&=&\int\e^{\eta_i+u_{0i}-\omega-\ln h_0}\,\dd\mu\nn\\
&\geq&\int\dd\mu \exp\left(\frac1{\int\dd\mu}\int(\eta_i+u_{0i}-w-\ln h_0)\,\dd\mu\right)\nn\\
&=&K_i \e^{\bar{\eta}_i},\quad i=1,2,
\eea
where $K_i$ ($i=1,2$) are some positive constants and no summation convention is applied here over repeated indices.

Now, in view of (\ref{x531})--(\ref{x533}), we see that (\ref{4.4}) renders the lower estimate
\bea\label{x534}
I(\eta_1,\eta_2)&\geq&\frac14\left(\min\left\{\frac1\gamma,1\right\}-\frac{\vep}{\gamma}\right)\left(\|\dot{\eta}'_1\|^2_{L^2(\dd x)}+\|\dot{\eta}'_2\|^2_{L^2(\dd x)}\right)\nn\\
&&+\lm\left(K_1\e^{\bar{\eta}_1}+K_2\e^{\bar{\eta}_2}\right)-\frac12(\kappa_1\bar{\eta}_1+\kappa_2\bar{\eta}_2)-\frac{C(\vep)}{4\gamma}.
\eea
Choose $\vep<\gamma\min\{\frac1{\gamma},1\}=\min\{\gamma,1\}$. Then the right-hand side of (\ref{x534}) is bounded from below. Hence $I_0\equiv\inf\{I(\eta_1,\eta_2)\,|\,\eta_1,\eta_2\in\HH\}$ is well defined. Let $\{\eta^{(n)}_1,\eta^{(n)}_2\}\in\HH$ such that $I(\eta^{(n)}_1,\eta^{(n)}_2)\to I_0$ as $n\to\infty$. Then the structure of
the right-hand side of (\ref{x534}) indicates that $\{\|(\dot{\eta}^{(n)}_i)'\|^2_{L^2(\dd x)}\}$ and $\{\bar{\eta}^{(n)}_i\}$ ($i=1,2$) are bounded
sequences. Hence, without loss of generality, we may assume $\eta^{(n)}_i\to\eta_i$ weakly in $\HH$ ($i=1,2$) as $n\to\infty$. By the weak semicontinuity of $I$
over $\HH$, we obtain $I(\eta_1,\eta_2)=I_0$. That is, $(\eta_1,\eta_2)$ solves (\ref{4.25}). The strict convexicity of $I$ then implies that such a solution as
a critical point of $I$ in $\HH$ is unique.

Let $(\eta_1,\eta_2)$ be the solution of (\ref{4.25}) just obtained. Then we have
\bea\label{535}
&&\int (\zeta_1',\zeta_2')\Gamma^{-1}(\eta_1',\eta_2')^\tau +\lm\int\left(\text e^{\eta_1+u_{01}-\omega} \zeta_1+\e^{\eta_2+u_{02}-\omega}\zeta_2\right)\nn\\
&&-\frac1{4\gamma}
\int\left\{\Big((1+\gamma)u_{01}''+(\gamma-1)u_{02}''\Big)\zeta_1+\Big((\gamma -1)u_{01}''+(1+\gamma)u_{02}''\Big)\zeta_2\right\}\nn\\
&&=0,\quad \forall \zeta_1,\zeta_2\in\HH.
\eea
Thus, we may insert $\zeta_1\equiv1,\zeta_2\equiv0$ and $\zeta_1\equiv0,\zeta_2\equiv1$ in (\ref{535}) to get 
\bea
\int \text e^{\eta_1+u_{01}-\omega} &=&\frac1{4\lm\gamma}
\int\Big((1+\gamma)u_{01}''+(\gamma-1)u_{02}''\Big) =\frac{\kappa_1}{2\lm},\label{536}\\
\int \text e^{\eta_2+u_{02}-\omega} &=&\frac1{4\lm\gamma}
\int\Big((1-\gamma)u_{01}''+(\gamma+1)u_{02}''\Big)=\frac{\kappa_2}{2\lm},\label{537}
\eea
respectively. In particular, we see that the condition $\kappa_1,\kappa_2>0$ rises also as a necessary condition, as a consequence of these relations.

Since $\eta'_1,\eta'_2\in L^2(\dd x)$, we see that, in the sense of subsequences at least, there holds $\eta'_1(x),\eta'_2(x)\to0$ as $x\to\pm\infty$.
Now denote the right-hand sides of (\ref{4.2a}) and (\ref{4.2b}) by $f_1(x)$ and $f_2(x)$, respectively. Then 
\be\label{538}
f_i(x)=\mbox{o}(\e^{-|x|}) \mbox{ (say)} \quad \mbox{as }|x|\to\infty,\quad i=1,2.
\ee
Thus, using (\ref{538}) and $|\eta'_i(x)|=|\int_x^{\pm\infty} f_i(y)\,\dd y|$ ($i=1,2$), we see that $\eta_i'(x)$ vanishes at $\pm\infty$ as fast, which implies
$\eta_i(x)\to$ constants as $x\to \pm\infty$, $i=1,2$.

Summarizing and returning to the original variables $q_1,q_2$, all the statements made in Theorem \ref{th4.1} are established.

We note that the direct minimization method used here adapts that developed in \cite{LY} for solving a system of multivortex equations
over a doubly periodic domain arising in a
supersymmetric gauge field theory model. An additional difficulty encountered is that we need to deal with the full real line, $\bfR$, where our equations
are sitting over. Fortunately, such a difficulty may effectively be overcome by a weighted Sobolev space formalism, allowing the execution of the direct
minimization method.

\section{Electroweak domain walls in the Ambj{\o}rn--Olesen situation}\label{s6}
\setcounter{equation}{0}

In this section we consider domain wall solitons arising in the formalism of Ambj{\o}rn--Olesen \cite{AO2,AO3,AO4} of
the  classical electroweak theory of
Weinberg--Salam \cite{Lai} governing the $W$-boson condensed vortices with a BPS structure. In this theory, the Higgs field $\phi$ is a complex doublet in the fundamental
representation of
 $SU(2)\times U(1)$ which transforms $\phi$
according to the rule
\be
%\begin{array}{rl}
\phi\mapsto \exp(-\ii\omega_a t_a)\phi,\quad \omega_a\in \bfR,\quad a=1,2,3,\quad
\phi\mapsto \exp(-\ii\xi t_0)\phi,\quad\xi\in \bfR, 
\ee
where $t_a=\frac{\sigma_a}{2}$ ($a=1,2,3$) and $t_0=\frac12{\bf 1}$
is a generator of $U(1)$ in the above matrix representation.
The $SU(2)$ and $U(1)$ gauge fields are denoted, respectively, by
 $A_\mu=A^a_\mu t_a$
 and $B_\mu$. 
The field strength tensors and the $SU(2)\times U(1)$ gauge-covariant derivative
are defined by 
\bea
%\begin{array}{rl}
F_{\mu\nu}&=&\partial_\mu A_\nu-\partial_\nu A_\mu+\ii g[A_\mu,A_\nu],\quad
H_{\mu\nu}=\partial_\mu B_\nu-\partial_\nu B_\nu,\nn\\
D_\mu\phi&=&\partial_\mu\phi+\ii gA^a_\mu t_a\phi+\ii g'B_\mu t_0\phi,
%\end{array}
\eea
where $g,g'>0$ are coupling constants.  
The Lagrangian density of the bosonic electroweak theory is then given as
\begin{equation}\label{f35}
{\cal L}=-\frac14(F^{a\mu\nu} F^a_{\mu\nu}+H^{\mu\nu}H_{\mu\nu})+(D^\mu\phi)^\dagger
\cdot(D_\mu\phi)-\Lambda(\vp_0^2-\phi^\dagger\phi)^2,
\end{equation}
where 
$\Lambda, \vp_0$ are positive parameters
with $\Lambda$ giving rise to the
Higgs particle mass and $\vp_0$ setting the energy scale of symmetry-breaking.
In such a context,
a pair of new vector fields $P_\mu$ and $Z_\mu$ arise resulting from
a rotation of the pair
 $A^3_\mu$ and
$B_\mu$,
\be
P_\mu=B_\mu\cos\theta+A^3_\mu\sin\theta,\quad
Z_\mu=-B_\mu\sin\theta+A^3_\mu\cos\theta. 
\ee
In terms of $P_\mu$ and $ Z_\mu$, the covariant derivative $D_\mu$ becomes
\be
D_\mu=
\partial_\mu+\ii g(A^1_\mu t_1+A^2_\mu t_2)
+\ii
P_\mu(g\sin\theta t_3+g'\cos\theta t_0)
+\ii Z_\mu(g\cos\theta t_3-g'\sin\theta t_0).
\ee
Requiring that the coefficient of $P_\mu$ be the charge operator $eQ=e(t_3+t_0)$ where
$-e$ is the charge of the electron, we obtain the relations
\be\label{f36}
e=g\sin\theta=g'\cos\theta
=\frac{gg'}{(g^2+g'^2)^{1/2}},\quad
\cos\theta=\frac{g}{(g^2+g'^2)^{1/2}},
\ee
which defines the
Weinberg mixing angle, $\theta$ ($\approx 30^\circ$).  Thus now $D_\mu$ takes the form
\be
D_\mu=\partial_\mu+\ii g(A^1_\mu t_1+A^2_\mu t_2)+\ii P_\mu eQ+\ii Z_\mu eQ',
\ee
where $Q'=\cot\theta t_3-\tan\theta t_0$ is the neutral charge operator.
From (\ref{f36}), when we impose the unitary gauge in which $\phi=(0,\vp)^\tau$ with
 $\vp$  a real scalar field, then
\be
D_\mu\phi=\left(
               \frac{\ii}2g(A^1_\mu-\ii A^2_\mu)\vp,
               \partial_\mu\vp-\frac {\ii g}{2\cos\theta}Z_\mu\vp\right)^\tau.
\ee
Define now the complex vector field
$
W_\mu=\frac1{\sqrt{2}}(A^1_\mu+\ii A^2_\mu)
$
and the covariant derivative ${\cal D}_\mu=\partial_\mu-\ii g A^3_\mu$.  
We see that, with the notation 
$
P_{\mu\nu}=\partial_\mu P_\nu-\partial_\nu P_\mu,
Z_{\mu\nu}=\partial_\mu Z_\nu-\partial_\nu Z_\mu,
$
 the Lagrangian (\ref{f35})
takes the form
\bea\label{f37}
{\cal L}&=&-\frac12\overline{(\D^\mu W^\nu-\D^\nu W^\mu)}(\D_\mu W_\nu-\D_\nu W_\mu)-\frac14
          Z^{\mu\nu}Z_{\mu\nu}-\frac14P^{\mu\nu}P_{\mu\nu}\nn\\
&&-\frac12g^2([W^\mu \overline{W}_\mu]^2-[W^\mu W_\mu]\overline{[W^\nu W_\nu]})\nn\\
          &&-\ii g(Z^{\mu\nu}\cos\theta+
P^{\mu\nu}\sin\theta)\overline{W}_\mu W_\nu\nn\\
&&+\frac12g^2\vp^2W^\mu \overline{W}_\mu+\partial^\mu\vp\partial_\mu\vp
          +\frac1{4\cos^2\theta}g^2\vp^2Z^\mu Z_\mu-\Lambda(\vp^2_0-\vp^2)^2.
\eea
Thus the theory is now  reformulated in the celebrated unitary gauge. The $W$ and $Z$ fields
 represent
two massive vector bosons which eliminate the curious massless
 Goldstone particles
in the original setting (\ref{f35}). These fields mediate short-range
(weak) interactions. The remaining massless gauge
(photon) field $P_\mu$ arising from the residual $U(1)$ symmetry mediates 
long-range (electromagnetic) interactions.
To proceed further, we assume that electroweak excitation is in the
third direction. 
Thus, we arrive at the vortex ansatz 
\be\label{f38}
A^a_0=A^a_3=B_0=B_3=0,\,
A^a_j=A^a_j(x^1,x^2),\,
B_j=B_j(x^1,x^2),\, j=1,2,\,
\phi=\phi(x^1,x^2). 
\ee
As a consequence, if the corresponding
$W_1$ and $W_2$ are represented by a complex scalar field   
$W$ according to $W_1=W, W_2=\ii W$
(this implies the relation
$A^1_2=-A^2_1\,,\,\,A^2_2=A^1_1$), 
the energy density associated with (\ref{f37})
takes the form \cite{AO2,AO3,AO4}
\bea\label{f39}
{\cal H}&=&
|\D_1W+\ii\D_2W|^2+\frac12P^2_{12}+
\frac12Z^2_{12}
-2g(Z_{12}\cos\theta+P_{12}\sin\theta)
                        |W|^2\nn\\
         &&+2g^2|W|^4+(\partial_j\vp)^2+\frac1{4\cos^2\theta}g^2\vp^2Z^2_j
+g^2\vp^2|W|^2
          +\Lambda (\vp^2_0-\vp^2)^2.
\eea
There is a residual $U(1)$ symmetry in the model which
may clearly be seen from the
 invariance of (\ref{f39})
under the gauge transformation   
\begin{equation}\label{f40}
W\mapsto\exp(\ii\zeta)W,\quad P_j\mapsto P_j+\frac1e\partial_j\zeta,\quad
 Z_j\mapsto Z_j,
\quad \vp\mapsto
\vp,  
\end{equation}
due to (\ref{f36}). The Euler--Lagrange equations of (\ref{f39}) are still complicated. In \cite{AO2,AO3,AO4}, it is shown that, there is a critical coupling
situation when
\be
\Lambda=\frac{g^2}{8\cos^2\theta},
\ee
there hold the BPS vortex equations 
\bea\label{xx5.1}
&&{\cal D}_1 W+\ii {\cal D}_2 W=0,\quad
P_{12}=\frac{g}{2\sin\theta}\varphi_0^2+2g\sin\theta|W|^2,\nn\\
&&Z_{12}=\frac g{2\cos\theta}(\varphi^2-\varphi_0^2)+2g\cos\theta|W|^2,\quad
Z_j=-\frac{2\cos\theta}g\epsilon_{jk}\pa_k\ln\varphi.
\eea
In \cite{SY1,SY2}, some existence theorems for the multivortex solutions of (\ref{xx5.1}) are obtained. See \cite{BM,BT,CL} for some further development.

We are  now ready to consider a domain-wall
substructure contained in the Ambj{\o}rn--Olesen theory. For this purpose, we take the consistent ansatz
in (\ref{xx5.1}) that $W=w$ is real-valued, all fields depend on $x^1=x$ only,  the $j=1$ components
of all the gauge fields vanish, and $P_2=P,Z_2=Z$. Thus we arrive at the reduced equations
\bea\label{xx5.3}
w'&=&-g(\sin\theta P+\cos\theta Z)w,\quad
\varphi'=\frac{g}{2\cos\theta}Z\varphi,\nn\\
P'
&=&\frac g{2\sin\theta} \varphi_0^2+2g\sin\theta w^2,\quad
Z'=\frac g{2\cos\theta}(\varphi^2-\varphi_0^2)+2g\cos\theta w^2.
\eea
The first two equations in (\ref{xx5.3}) indicate that the fields $w,\varphi$ will stay positive once they are positive, which enables us to assume $w,\varphi>0$. Therefore, we see that the system (\ref{xx5.3}) leads us to the second-order equations
\be\label{6.16}
(\ln w)''=-\frac{g^2}2\varphi^2-2g^2 w^2,\quad
(\ln \varphi)''=\frac{g^2}{4\cos^2\theta}(\varphi^2-\varphi_0^2)+g^2 w^2.
\ee

In view of the study of the electroweak vortices over $\bfR^2$ carried out in \cite{SY2} and the discussion of Sections 4--5, we impose
the boundary condition of
an ``instanton" lump such that $w,\vp=0$ at infinity.
For such solutions, we may state our main existence theorem as follows.

\begin{theorem} \label{th2}
For any coupling parameters $g,\vp_0>0, 0<\theta<\frac{\pi}2$ and any prescribed parameters $\alpha_1,\beta_1, \alpha_2<0,\beta_2<0$
satisfying
\be\label{xx6.17}
\alpha_1+\beta_1>0,\quad   |\alpha_2|+|\beta_2|>\frac{\alpha_1+\beta_1}{\tan^2\theta},\quad \min\{|\alpha_2|,|\beta_2|\}+\frac{\alpha_1+\beta_1}{\tan^2\theta}>(|\alpha_2|+|\beta_2|),
\ee
we can construct a solution $(w,\vp,P,Z)$ of (\ref{xx5.3}) through solving the equivalent system of the second-order equations (\ref{6.16}) subject to the boundary condition
$w(x),\vp(x)\to 0$ as $x\to\pm\infty$, with the fulfillment of the following precise properties:
\begin{enumerate}
\item[(i)] The functions $w$ and $\vp$ possess the sharp decay behavior
\bea
w^2(x)&=&\mbox{\rm O}\left(\e^{-|\alpha_2|x}\right),\quad x\to\infty,\quad  w^2(x)=\mbox{\rm O}\left(\e^{|\beta_2|x}\right),\quad x\to-\infty,\\
\vp^2(x)&=&\mbox{\rm O}\left(\exp\left\{-\frac{g^2\vp_0^2}{4\cos^2\theta}x^2+\frac12(\alpha_1+|\alpha_2|)x\right\}\right),\quad x\to\infty,\\
\vp^2(x)&=&\mbox{\rm O}\left(\exp\left\{-\frac{g^2\vp_0^2}{4\cos^2\theta}x^2-\frac12(\beta_1+|\beta_2|)x\right\}\right),\quad x\to-\infty.
\eea
\item[(ii)] There hold the exact total  integrals
\be
\int w^2(x)=\frac1{4g^2}\left(|\alpha_2|+|\beta_2|-\frac{\alpha_1+\beta_1}{\tan^2\theta}\right),\quad \int\vp^2(x)=\frac{\alpha_1+\beta_1}{g^2\tan^2\theta}.
\ee
\end{enumerate}
\end{theorem}

Note that, in (\ref{xx6.17}), we only need to require the joint condition $\alpha_1+\beta_1>0$ but we do not 
need to require $\alpha_1,\beta_1$ be positive individually.

A sketch of the proof of the above theorem will be presented in the next section.

\section{Construction of electroweak domain wall solitons}
\setcounter{equation}{0}

In this section, we demonstrate the key steps in the proof of Theorem \ref{th2}.

Setting $v_1=2\ln w$ and $v_2=2\ln \left(\frac{\varphi}{\varphi_0}\right)$ in (\ref{6.16}), we obtain
\be\label{6.17}
v_1''=-4g^2\e^{v_1}-{g^2}{\varphi_0^2} \e^{v_2},\quad
v_2''=2 g^2\e^{v_1}+\frac{g^2\varphi^2_0}{2\cos^2\theta}(\e^{v_2}-1).
\ee
So the instanton lump type boundary condition, $w,\vp=0$ at infinity, as stated in Theorem \ref{th2}, becomes
\be\label{6.6}
v_1(\pm\infty)=-\infty,\quad v_2(\pm\infty)=-\infty.
\ee

As in Section 5, we will  pursue a variational solution of the problem. However, the coefficient matrix of the nonlinear terms of the equations
in (\ref{6.17}) is not symmetric, so that the formalism we used earlier is not applicable here. To overcome this difficulty, we use the new variables
\be\label{6.7}
u_1=v_1+2v_2,\quad u_2=v_1,
\ee
to recast (\ref{6.17}) into
\be\label{x6.20}
u_1''=g^2\varphi_0^2\tan^2\theta\e^{\frac12(u_1-u_2)}-\frac{g^2\varphi_0^2}{\cos^2\theta},\quad
u_2''=-g^2\varphi_0^2\e^{\frac12(u_1-u_2)}-4g^2\e^{u_2}.
\ee
Setting $\lm=\frac{g^2\varphi_0^2}{2\cos^2\theta}$ and adopting the same notation as in (\ref{4.1}) for the background functions $u_{0i}$ ($i=1,2$) and 
$\omega$, we see that, with $u_1=\eta_1+u_{01}-\omega$, $u_2=\eta_2+u_{02}$, the equations in  \eqref{x6.20} become
\bea
\eta_1''&=&g^2\varphi_0^2\tan^2\theta U\e^{\frac12(\eta_1-\eta_2)}-u_{01}'',\label{x6.21}\\
\eta_2''&=&-g^2\varphi_0^2 U\e^{\frac12(\eta_1-\eta_2)}-4g^2V\e^{\eta_2}-u_{02}'',\label{x6.22}
\eea
where
\be
U=\exp\left(\frac12(u_{01}-u_{02}-\omega)\right)\quad\mbox{and}\quad V=\e^{u_{02}}
\ee
are two weight functions. In order to ensure their exponential decay, we need to impose $\alpha_2,\beta_2<0$ in (\ref{4.1}) throughout this section.

Formerly, it is seen that the equations (\ref{x6.21}) and (\ref{x6.22}) permit the variational functional
\bea\label{x6.24}
&&E(\eta_1,\eta_2)\nn\\
&&=\int\left(\frac1{2\tan^2\theta}(\eta_1')^2+\frac{1}2(\eta_2')^2+2g^2\vp_0^2 U\e^{\frac12(\eta_1-\eta_2)}-4g^2V\e^{\eta_2}-\frac1{\tan^2\theta}u''_{01}\eta_1
-u''_{02}\eta_2\right)\!,\quad\quad\quad
\eea
over the natural space $\HH$, say, defined in Section 5. Unfortunately, this functional is not bounded from below, as may be seen by taking $\eta_1,\eta_2=$ constants as testing functions. Thus a direct minimization approach as used in Section 5 is not available.

To overcome this difficulty, we shall pursue a constrained minimization method as in \cite{SY2}. For this purpose, we first note that, if $(\eta_1,\eta_2)$ is a critical
point of (\ref{x6.24}) in $\HH$, then there hold
\bea
&&\int\left(\eta_1'\zeta_1'+g^2\vp_0^2\tan^2\theta U\e^{\frac12(\eta_1-\eta_2)}\zeta_1\right)=\int u_{01}''\zeta_1,\quad\forall\zeta_1\in\HH,\label{x6.25}\\
&&\int\left(\eta_2'\zeta_2'-g^2\vp_0^2 U \e^{\frac12(u_1-u_2)}\zeta_2-4g^2V\e^{\eta_2}\zeta_2\right)=\int u''_{02}\zeta_2,\quad\forall\zeta_2\in\HH.\label{x6.26}
\eea
Consequently, taking $\zeta_1,\zeta_2\equiv1$ as testing functions in (\ref{x6.25}), (\ref{x6.26}), respectively, we arrive at the constraints
\bea
&&\int g^2\varphi_0^2\tan^2\theta U\e^{\frac12(\eta_1-\eta_2)}=\alpha_1+\beta_1>0,\label{x6.27} \\
&&\int \left( g^2\varphi_0^2 U\e^{\frac12(\eta_1-\eta_2)}+4g^2V\e^{\eta_2}\right)=-(\alpha_2+\beta_2)=|\alpha_2|+|\beta_2|,\label{x6.28}
\eea
which are equivalent to saying that the right-hand sides of (\ref{x6.21}) and (\ref{x6.22}) are both integrated to zero value.

In view of (\ref{x6.27}) and (\ref{x6.28}), we obtain the following necessary condition on the free parameters, in addition to (\ref{x6.27}), 
\be\label{x6.29}
\int 4g^2V\e^{\eta_2}=(|\alpha_2|+|\beta_2|)-\frac{(\alpha_1+\beta_1)}{\tan^2\theta}>0.
\ee

From (\ref{x6.27}) and (\ref{x6.29}), we see that the exponential terms in (\ref{x6.24}) may be viewed as  ``frozen" so that we may consider all the ``active"
terms in it which may collectively be put into a new functional
\be\label{x6.30}
I(\eta_1,\eta_2)
=\int\left(\frac1{2\tan^2\theta}(\eta_1')^2+\frac{1}2(\eta_2')^2-\frac1{\tan^2\theta}u''_{01}\eta_1
-u''_{02}\eta_2\right).
\ee

We now show that a critical point of (\ref{x6.30}) subject to the constraints (\ref{x6.27}) and (\ref{x6.29}) is a solution to the equations (\ref{x6.21}) and (\ref{x6.22}). That is, we show that the Lagrange multipliers arise from such formulated constrained minimization problem yield correct values to allow us to
recover the equations. To this end, notice that the functionals on the left-hand sides of (\ref{x6.27}) and (\ref{x6.29}) have linearly independent Fr\'{e}chet
derivatives. Hence there are numbers (the renormalized Lagrange multipliers), $\xi_1,\xi_2\in\bfR$, such that
\bea
&&\frac1{\tan^2\theta}\int \left(\eta_1'\zeta_1'-u_{01}''\zeta_1\right)
=\xi_1\int g^2\varphi_0^2 U\e^{\frac12(\eta_1-\eta_2)}\zeta_1,\quad \zeta_1\in\mathscr{H},\label{x6.31}\\
&&\int\left(\eta_2'\zeta_2'- u_{02}''\zeta_2\right)
=-\xi_1\int_{\mathbb{R}}g^2\varphi_0^2 U\e^{\frac12(\eta_1-\eta_2)}\zeta_2+\xi_2\int g^2V\e^{\eta_2}\zeta_2,\quad\zeta_2\in\mathscr{H}.\label{x6.32}
\eea
Setting $\zeta_{1,2}\equiv1$ in the above equations and applying (\ref{x6.27}), (\ref{x6.29}), we obtain
\be\label{x6.33}
\xi_1=-1,\quad \xi_2=4.
\ee
Inserting (\ref{x6.33}) into (\ref{x6.31}) and (\ref{x6.32}), we see that the equations (\ref{x6.31}) and (\ref{x6.32}) give the weak form of the coupled equations
(\ref{x6.21}) and (\ref{x6.22}) as desired.

We next show the existence of a minimizer of the functional (\ref{x6.30}) in $\HH$ subject to (\ref{x6.27}) and (\ref{x6.29}) where $\HH$ is as defined in Section 5 with the norm given by (\ref{xx5.10})
for which the weight exponent $\beta$ satisfies
\be\label{x6.34}
0<\beta<\min\{|\alpha_2|,|\beta_2|\}.
\ee
This condition ensures that the functional given on the left-hand side of (\ref{x6.29}) is continuous with respect to the weak topology of $\HH$. Note
that the functional given
by the left-hand side of (\ref{x6.27}) automatically enjoys such a property for all values of the parameters since the weight function $U$ decays like
$\e^{-\frac\lm2 x^2}$ near infinities.

With the notation of Section 5, for $\eta_i=\overline{\eta}_i+\dot{\eta}_i$ ($i=1,2$), we may resolve (\ref{x6.27}) and (\ref{x6.29}) to get
\be\label{x6.35}
\overline{\eta}_1-\overline{\eta}_2=2\ln\gamma_1-2\ln\int U\e^{\frac12(\dot{\eta}_1-\dot{\eta}_2)},\quad
\overline{\eta}_2=\ln\gamma_2-\ln\int V\e^{\dot{\eta}_2},
\ee
where we use the suppressed notation
\be
\gamma_1=\frac{\alpha_1+\beta_1}{g^2\vp_0^2\tan^2\theta},\quad \gamma_2=\frac1{4g^2}\left(|\alpha_2|+|\beta_2|-\frac{(\alpha_1+\beta_1)}{\tan^2\theta}\right).
\ee
Substituting (\ref{x6.35}) into (\ref{x6.30}), we have
\bea\label{x6.37}
&&I(\eta_1,\eta_2)
=\int\left(\frac1{2\tan^2\theta}(\eta_1')^2+\frac{1}2(\eta_2')^2-\frac1{\tan^2\theta}u''_{01}\dot{\eta}_1
-u''_{02}\dot{\eta}_2\right)\nn\\
&&+\frac{2(\alpha_1+\beta_1)}{\tan^2\theta}\ln\int U\e^{\frac12(\dot{\eta}_1-\dot{\eta}_2)}
-\left(|\alpha_2|+|\beta_2|-\frac{\alpha_1+\beta_1}{\tan^2\theta}\right)\ln\int V\e^{\dot{\eta}_2}\nn\\
&&-\frac{(\alpha_1+\beta_2)}{\tan^2\theta}(2\ln\gamma_1+\ln\gamma_2)+\left(|\alpha_2|+|\beta_2|\right)\ln\gamma_2.
\eea
In view of the Poincar\'{e} inequality (\ref{x5.25}), the dotted terms in the first line of (\ref{x6.37}) are well controlled.
In view of Jensen's inequality with respect to the weighted measure $\dd\mu$, the first term in the second line of (\ref{x6.37}) is bounded from below. For the second term, we may use the Schwarz inequality and (\ref{x5.13}) to get
\be\label{x6.38}
\ln \int V\e^{\dot{\eta}_2}=\ln\int V{h_0^{-1}}\e^{\dot{\eta}_2}\,\dd\mu\leq \frac1{2b}\int (\dot{\eta}'_2)^2 \dd x+C_b,\quad \forall b\in(0,\beta),
\ee
where $C_b>0$ depends on $b$. Thus, in (\ref{x6.37}), we have by (\ref{x6.38}) the lower bound
\be\label{x6.39}
\int\frac12(\eta_2')^2-\left(|\alpha_2|+|\beta_2|-\frac{\alpha_1+\beta_1}{\tan^2\theta}\right)\ln\int V\e^{\dot{\eta}_2}
\geq\frac1{2b}\left\{b-\left(|\alpha_2|+|\beta_2|-\frac{\alpha_1+\beta_1}{\tan^2\theta}\right)\right\}\int(\dot{\eta}'_2)^2.
\ee
Consequently, if we impose the condition
\be\label{x6.40}
\min\{|\alpha_2|,|\beta_2|\}>|\alpha_2|+|\beta_2|-\frac{\alpha_1+\beta_1}{\tan^2\theta},
\ee
then we may choose $\beta$ to satisfy (\ref{x6.34}) and $b\in(0,\beta)$ so that the factor in front of the integral on the right-hand side of (\ref{x6.39}) is positive.
These lead us to the coercive lower bound
\be
I(\eta_1,\eta_2)\geq C_1 \|\dot{\eta}'_1\|^2_{L^2(\dd x)}+C_2\|\dot{\eta}'_2\|^2_{L^2(\dd x)}-C_3,\quad \eta_1,\eta_2\in\HH,
\ee
where $C_1,C_2,C_3$ are some positive constants independent of $\eta_1,\eta_2$. Hence, the existence of a minimizer of $I$ in $\HH$,
say $(\eta_1,\eta_2)$, subject to the constraints
(\ref{x6.27}) and (\ref{x6.29}), follows. 

We have seen that $(\eta_1,\eta_2)$ is a solution of the equations (\ref{x6.21}) and (\ref{x6.22}).
Using the same method as in Section 5, we may show that $\eta_i(x)\to$ constants as $x\to\pm\infty$ for $i=1,2$. Therefore we obtain the sharp asymptotic
estimates for $u_1=\eta_1+u_{01}-\omega$ and $u_2=\eta_2+u_{02}$ as follows:
\bea
&&u_1(x)=\mbox{O}(\alpha_1 x -\lm x^2),\quad u_2(x)=\mbox{O}(-|\alpha_2|x),\quad x\to\infty, \\
&&u_1(x)=\mbox{O}(-\beta_1 x -\lm x^2),\quad u_2(x)=\mbox{O}(|\beta_2|x),\quad x\to-\infty.
\eea

Finally, returning to the original variables through $w^2=\e^{v_1},\vp^2=\vp_0^2\e^{v_2}$, and (\ref{6.7}), we see that all the statements made in
Theorem \ref{th2} are established.

\section{Conclusions}

In this work, we have presented a series of BPS domain wall solitons arising in several classical gauge field theories, both Abelian and non-Abelian. In the former situation under consideration,
the governing equations are either partially or completely integrable to allow a thorough description of their relevant solutions. In the latter situation, 
the governing equations are complicated but may be investigated using techniques from calculus of variations which provide us with rather rich families of solutions
depending on several prescribed parameters. Included in our results are the following.

\begin{enumerate}
\item[(i)] For the BPS domain wall equations derived in \cite{Bol} in the Abelian Higgs theory, the governing equation may be integrated once to be reduced into a Friedmann type equation.
Although this reduced equation cannot be integrated further, a direct analysis may be carried out to give us a complete understanding of all the relevant
solutions realizing various phase-transition boundary conditions.

\item[(ii)] For the nonrelativistic and relativistic Abelian Chern--Simons--Higgs vortex equations, derived in \cite{JP1,JP2} and \cite{HKP,JW}, respectively, we have shown
that there are substructures to allow domain wall equations as that in the Abelian Higgs theory. Furthermore, we have shown that these equations are all integrable to allow all the relevant solutions to be
constructed explicitly. As by-products, with respect to integrability, we see  through the domain wall equation substructures that the Chern--Simons--Higgs equations are more
integrable than the Abelian Higgs equations; with respect to structures of solutions, we acquire a complete knowledge of solutions that interpolate topological
and non-topological boundary conditions, which may be useful and suggestive to the study of the BPS vortex equations in their original settings.

\item[(iii)] For the BPS domain wall equations arising in the $U(2)$ Yang--Mills--Higgs theory in which the Higgs scalar is taken to stay in the matrix representation,
derived in \cite{Bol}, although the existence of solutions realizing relevant boundary conditions may be achieved by reducing the system of equations into a
single equation already studied in (i) above, we have shown that, at least in the magnetic-to-magnetic phase transition situation, the system possesses a
$4$-parameter family of solutions which cannot be generated by the solution of the reduced single equation which actually depends on 2 parameters. Thus, in 
particular, the system of equations is not equivalent to its single equation reduction. Our method is through a direct minimization of the associated energy/action
functional for the system of equations. We have developed a weighted Sobolev space formalism which is useful for tackling other complicated domain wall 
equation problems.

\item[(iv)] For the BPS vortex equations \cite{AO2,AO3,AO4} arising in the classical electroweak theory where the Higgs field stays in the fundamental 
representation of $SU(2)\times U(1)$, we have shown that there is again a substructure of domain wall equations. However, unlike in (iii) above, the
coefficient matrix of the nonlinear terms is not symmetric so that the direct minimization method fails. In fact, after a suitable change of variables, 
it can be seen that the equations
enjoy a variational principle. Unfortunately, the energy functional is not bounded from below. In order to overcome this difficulty, we have formulated
a multiply constrained minimization approach. In such a setting, the constraints serve the purpose of making the energy functional bounded from below, or
more strongly, coercive, which enables us to prove the existence of a minimizer. Since the constraints give rise to the Lagrange multipliers, 
in order that our method works, we must show that
their values are precisely those needed to recover the equations of motion. We have shown indeed this is as desired. As a result, we are able to construct
again a $4$-parameter family of domain wall solutions under some explicitly stated sufficient conditions imposed on these parameters. This part of the work
indicates that the weighted Sobolev space formalism developed in (iii) is useful for treating other domain wall equations with different technical challenges.

\end{enumerate}

\medskip

{\bf Acknowledgments.} 
YY was partially supported
by Natural Science Foundation of China under Grant No. 11471100.

\end{document}